\documentclass[10pt, conference, letterpaper]{IEEEtran}
\usepackage{amstext}
\usepackage{amsbsy}
\usepackage{amsthm}
\usepackage{amsmath}
\usepackage{amssymb}
\usepackage{bm, upgreek}
\usepackage{color}
\usepackage{scalefnt}
\usepackage{setspace}

\usepackage{enumitem}
\usepackage{url}
\usepackage[numbers,sort&compress]{natbib}
\usepackage[linesnumbered, ruled, vlined]{algorithm2e}
\usepackage{graphicx}
\usepackage{float}
\usepackage{epstopdf}
\usepackage{comment} 

\newtheorem{myDef}{Definition}
\newtheorem{myTheo}{Theorem}
\newtheorem{myLemma}{Lemma}
\newtheorem{myCor}{Corollary}
\newtheorem{myObs}{Observation}

\setlist{topsep=0pt}
\newcommand{\compresslist}{
  \setlength{\itemsep}{1pt}
  \setlength{\parskip}{0pt}
  \setlength{\parsep}{0pt}
  \setlength{\topskip}{0pt}
}

\renewenvironment{thebibliography}[1]{%
 \begin{oldthebibliography}{#1}%
   \setlength{\parskip}{0ex}%
   \setlength{\itemsep}{0ex}%
}%
{%
  \end{oldthebibliography}%
}

\usepackage{tikz}
\newcommand*\circled[1]{\tikz[baseline=(char.base)]{
            \node[shape=circle,draw,inner sep=0.25pt] (char) {#1};}}

\makeatletter
\def\thm@space@setup{%
  \thm@preskip=1pt
  \thm@postskip=\thm@preskip 
}
\makeatother
\allowdisplaybreaks
\ifCLASSINFOpdf
\else
\fi
\IEEEoverridecommandlockouts

\begin{document}
\title{CENTURION: Incentivizing Multi-Requester Mobile Crowd Sensing\thanks{We sincerely thank Professor Julia Chuzhoy for her valuable contribution. We gratefully acknowledge the support of National Science Foundation grants CNS-1330491, and 1566374. The views and conclusions contained in this document are those of the authors and should not be interpreted as necessarily representing the official policies, either expressed or implied, of the sponsors.}}
\author{Haiming Jin\IEEEauthorrefmark{1}, Lu Su\IEEEauthorrefmark{2}, Klara Nahrstedt\IEEEauthorrefmark{1}\\
\IEEEauthorblockA{
\IEEEauthorrefmark{1}Department of Computer Science, University of Illinois at Urbana-Champaign, IL, USA\\ 
\IEEEauthorrefmark{2}Department of Computer Science and Engineering, State University of New York at Buffalo, NY, USA\\
Email: hjin8@illinois.edu, lusu@buffalo.edu, klara@illinois.edu}
}
\maketitle

\begin{abstract}
The recent proliferation of increasingly capable mobile devices has given rise to mobile crowd sensing (MCS) systems that outsource the collection of sensory data to a crowd of participating workers that carry various mobile devices. Aware of the paramount importance of effectively \textit{incentivizing participation} in such systems, the research community has proposed a wide variety of incentive mechanisms. However, different from most of these existing mechanisms which assume the existence of only one data requester, we consider MCS systems with \textit{multiple data requesters}, which are actually more common in practice. Specifically, our incentive mechanism is based on \textit{double auction}, and is able to stimulate the participation of both data requesters and workers. In real practice, the incentive mechanism is typically not an isolated module, but  interacts with the \textit{data aggregation mechanism} that aggregates workers' data. For this reason, we propose CENTURION, a novel \textit{integrated framework} for \textit{multi-requester} MCS systems, consisting of the aforementioned incentive and data aggregation mechanism. CENTURION's incentive mechanism satisfies \textit{truthfulness}, \textit{individual rationality}, \textit{computational efficiency}, as well as guaranteeing \textit{non-negative social welfare}, and its data aggregation mechanism generates \textit{highly accurate} aggregated results. The desirable properties of CENTURION are validated through both theoretical analysis and extensive simulations. 
\end{abstract}
\IEEEpeerreviewmaketitle
\section{Introduction}

Recent years have witnessed the rise of mobile crowd sensing (MCS), a newly-emerged sensing paradigm that outsources the collection of sensory data to a crowd of participating users, namely (crowd) workers, who usually carry increasingly capable mobile devices (e.g., smartphones, smartwatches, smartglasses) with a plethora of on-board sensors (e.g., gyroscope, camera, GPS, compass, accelerometer). Currently, a large variety of MCS systems \cite{PMohanSenSys08, AThiagarajanSenSys09, JErikssonMobiSys08, myheartmap, SHuTOSN15, YChengSenSys14} have been deployed that cover almost every aspect of our lives, including healthcare, smart transportation, environmental monitoring, and many others. 

To perform the sensing tasks, the participating workers typically consume their own resources such as computing and communicating energy, and expose themselves to potential privacy threats by sharing their personal data. For this reason, a participant would not be interested in participating in the sensing tasks, unless she receives a satisfying reward to compensate her resource consumption and potential privacy breach. Therefore, it is necessary to design an effective \textit{incentive mechanism} that can achieve the maximum user participation. Due to the paramount importance of stimulating participation, many incentive mechanisms \cite{HKaiMobiHoc16, XHongAllerton14, CManHonMobiHoc15, GLinINFOCOM15, JHaimingMobiHoc15, YWenTVT14, ZDongINFOCOM14, CYanjiaoINFOCOM16, HShiboINFOCOM14, TLuoINFOCOM15, DLingjieINFOCOM12, ZQiINFOCOM15, ZXiangINFOCOM15, YDejunMobicom12, ZXinglinTPDS14, ZHonggangINFOCOM16, WYuemingINFOCOM15, IKoutsopoulosINFOCOM13, JHaimingMobiHoc16, JHaimingICDCS16, WJingICDCS16, FZhenniINFOCOM14, KMerkouriosMobiHoc16, PLingjunINFOCOM16, DPengMobiHoc15} have been proposed by the research community. However, most of these aforementioned past literature assume that there is only one data requester who also serves as the platform in the MCS system. In practice, however, there are usually \textit{multiple data requesters} competing for human resources, who usually outsource worker recruiting to third-party platforms (e.g., Amazon Mechanical Turk \cite{AMT}) that have already gathered a large number of workers. Therefore, in this paper, we focus on such MCS systems where three parties, including the data requesters, a platform (i.e., a cloud-based central server), as well as a crowd of participating workers co-exist, and aim to develop \textit{a new incentive mechanism that can decide which worker serves which data requester at what price}.

In real practice, the sensory data provided by individual workers are usually quite unreliable due to various factors (e.g., poor sensor quality, lack of sensor calibration, environment noise). Hence, in order to cancel out the possible errors from individual workers, it is highly necessary that the platform utilizes a \textit{data aggregation mechanism} to properly aggregate their noisy and even conflicting data. In an MCS system, the incentive and the data aggregation mechanism are usually not isolated from each other. In fact, the data aggregation mechanism typically interacts with the incentive mechanism, and thus, affects its design and performance. Intuitively, if the platform aggregates workers' data in naive ways (e.g., voting and average) that treat all workers' data equally, the incentive mechanism does not need to distinguish them with respect to their reliability. However, a weighted aggregation method that puts higher weights on more reliable workers is much more desirable, because it shifts the aggregated results towards the data provided by the workers with higher reliability. Accordingly, the incentive mechanism should also incorporate workers' reliability, and selects workers that are more likely to provide reliable data.

Therefore, different from most of the aforementioned existing work \cite{HKaiMobiHoc16, XHongAllerton14, CManHonMobiHoc15, GLinINFOCOM15, JHaimingMobiHoc15, YWenTVT14, ZDongINFOCOM14, CYanjiaoINFOCOM16, HShiboINFOCOM14, TLuoINFOCOM15, DLingjieINFOCOM12, ZQiINFOCOM15, ZXiangINFOCOM15, YDejunMobicom12, ZXinglinTPDS14, ZHonggangINFOCOM16, WYuemingINFOCOM15, IKoutsopoulosINFOCOM13, JHaimingMobiHoc16, JHaimingICDCS16, WJingICDCS16, FZhenniINFOCOM14, KMerkouriosMobiHoc16, PLingjunINFOCOM16, DPengMobiHoc15}, we propose CENTURION\footnote{The name CENTURION comes from in\underline{CENT}ivizing m\underline{U}lti-\underline{R}equester mob\underline{I}le cr\underline{O}wd se\underline{N}sing.}, a novel \textit{integrated framework} for \textit{multi-requester} MCS systems, which consists of a \textit{weighted data aggregation mechanism} that considers workers' diverse reliability in the calculation of the aggregated results, together with an incentive mechanism that selects workers who potentially will provide more reliable data. Specifically, CENTURION's incentive mechanism is based on \textit{double auction} \cite{RMcAfeeJET92}, which involves auctions among not only the workers, but also the data requesters, and is able to incentivize the participation of both data requesters and workers. This paper makes the following contributions. 
\begin{itemize}[leftmargin=*]\compresslist
\item Different from existing work, we propose a novel \textit{integrated framework} for \textit{multi-requester} MCS systems, called CENTURION, consisting of a data aggregation and an incentive mechanism. Such an integrated design, which captures the interactive effects between the two mechanisms, is much more complicated and challenging than designing them separately.  
\item CENTURION's double auction-based incentive mechanism is able to incentivize the participation of both data requesters and workers, and bears many desirable properties, including \textit{truthfulness}, \textit{individual rationality}, \textit{computational efficiency}, as well as \textit{non-negative social welfare}. 
\item The data aggregation mechanism of CENTURION takes into consideration workers' reliability, and calculates \textit{highly accurate} aggregated results. 
\end{itemize}

In the rest of this paper, we first discuss the past literature that are related to this work in Section \ref{sec:relatedwork}, and introduce the preliminaries in Section \ref{sec:prelim}. Then, the design details of CENTURION's data aggregation and incentive mechanism are described in Section \ref{sec:dd}. In Section \ref{sec:perleval}, we conduct extensive simulations to validate the desirable properties of CENTURION. Finally in Section \ref{sec:conc}, we conclude this paper.

\section{Related Work}\label{sec:relatedwork}
Aware of the paramount importance of attracting worker participation, the research community has recently developed various incentive mechanisms \cite{HKaiMobiHoc16, XHongAllerton14, CManHonMobiHoc15, GLinINFOCOM15, JHaimingMobiHoc15, YWenTVT14, ZDongINFOCOM14, CYanjiaoINFOCOM16, HShiboINFOCOM14, TLuoINFOCOM15, DLingjieINFOCOM12, ZQiINFOCOM15, ZXiangINFOCOM15, YDejunMobicom12, ZXinglinTPDS14, ZHonggangINFOCOM16, WYuemingINFOCOM15, IKoutsopoulosINFOCOM13, JHaimingMobiHoc16, JHaimingICDCS16, WJingICDCS16, FZhenniINFOCOM14, KMerkouriosMobiHoc16, PLingjunINFOCOM16, DPengMobiHoc15} for MCS systems. Among them, game-theoretic incentive mechanisms \cite{HKaiMobiHoc16, XHongAllerton14, CManHonMobiHoc15, GLinINFOCOM15, JHaimingMobiHoc15, YWenTVT14, ZDongINFOCOM14, CYanjiaoINFOCOM16, HShiboINFOCOM14, TLuoINFOCOM15, DLingjieINFOCOM12, ZQiINFOCOM15, ZXiangINFOCOM15, YDejunMobicom12, ZXinglinTPDS14, ZHonggangINFOCOM16, WYuemingINFOCOM15, IKoutsopoulosINFOCOM13, JHaimingMobiHoc16, JHaimingICDCS16, WJingICDCS16, FZhenniINFOCOM14}, which utilize either auction \cite{GLinINFOCOM15, ZQiINFOCOM15, ZXiangINFOCOM15, YDejunMobicom12, FZhenniINFOCOM14, JHaimingMobiHoc15, YWenTVT14, ZDongINFOCOM14, IKoutsopoulosINFOCOM13, ZXinglinTPDS14, ZHonggangINFOCOM16, WYuemingINFOCOM15, JHaimingMobiHoc16, JHaimingICDCS16, WJingICDCS16} or other game-theoretic models \cite{TLuoINFOCOM15, DLingjieINFOCOM12, CManHonMobiHoc15, XHongAllerton14, HShiboINFOCOM14, CYanjiaoINFOCOM16}, have gained increasing popularity due to their ability to tackle workers' selfish and strategic behaviors.  These mechanisms typically aim to maximize the platform's profit \cite{ZQiINFOCOM15, ZXiangINFOCOM15, YDejunMobicom12, ZXinglinTPDS14, TLuoINFOCOM15, DLingjieINFOCOM12, HShiboINFOCOM14, CYanjiaoINFOCOM16, ZHonggangINFOCOM16, WYuemingINFOCOM15} or social welfare \cite{GLinINFOCOM15, JHaimingMobiHoc15, YWenTVT14, ZDongINFOCOM14, CManHonMobiHoc15}, and minimize the platform's payment \cite{IKoutsopoulosINFOCOM13, XHongAllerton14, HKaiMobiHoc16, JHaimingMobiHoc16, JHaimingICDCS16, WJingICDCS16} or social cost \cite{FZhenniINFOCOM14}.

Different from most of the aforementioned past literature which assume that there exists only one data requester, we propose a novel incentive mechanism for MCS systems with \textit{multiple data requesters} that compete for human resources. In fact, there do exist several prior work \cite{ZHonggangINFOCOM16, ZXiangINFOCOM15, FZhenniINFOCOM14} designing incentive mechanisms for the multi-requester scenario. However, they do not provide any \textit{joint design} of the data aggregation and the incentive mechanism as in this paper, which is much more challenging than designing the two mechanisms as isolated modules. Moreover, although similar integrated designs that consider the two mechanisms are proposed in some existing work \cite{JHaimingMobiHoc16, JHaimingICDCS16}, as previously mentioned, they assume that only one data requester exists in the MCS system.

\section{Preliminaries}\label{sec:prelim}
In this section, we introduce the system overview, reliability level model, auction model, as well as the design objectives. 
\subsection{System Overview}\label{sec:sysoverview}
CENTURION is an MCS system framework consisting of a cloud-based platform, a set of participating workers, denoted as $\mathcal{W}=\{w_1,\cdots, w_N\}$, and a set of requesters, denoted as $\mathcal{R}=\{r_1,\cdots, r_M\}$. Each requester $r_j\in\mathcal{R}$ has a sensing task $\tau_j$ to be executed by the workers. The set of all requesters' tasks is denoted as $\mathcal{T}=\{\tau_1,\cdots, \tau_M\}$. We are specifically interested in the scenario where $\mathcal{T}$ is a set of $M$ different \textit{binary classification tasks} that require workers to locally decide the classes of the events or objects, and report to the platform their local decisions (i.e., the labels of the observed events or objects). Such MCS systems, collecting binary labels from the crowd, constitute a large portion of the currently deployed MCS systems (e.g., congestion detection systems that decide whether or not particular road segments are congested \cite{AThiagarajanSenSys09}, geotagging campaigns that tag whether bumps or potholes exist on specific segments of road surface \cite{PMohanSenSys08, JErikssonMobiSys08}). 

Each task $\tau_j$ has a true label $l_j\in\{-1,+1\}$, unknown to the requesters, the platform, and the workers. If a worker $w_i$ is chosen to execute task $\tau_j$, she will provide to the platform a label $l_{i,j}$. We define $\mathbf{l}=[l_{i,j}]\in\{-1,+1,\bot\}^{N\times M}$ as the matrix containing all workers' labels, where $l_{i,j}=\bot$ means that task $\tau_j$ is not executed by worker $w_i$. For every task $\tau_j$, the platform aggregates workers' labels into an aggregated result, denoted as $\widehat{l}_j$, so as to cancel out the errors from individual workers. The framework of CENTURION is given in Figure \ref{fig:sysframe}, and we describe its workflow as follows. 

\begin{figure}[htb]
\centering
\includegraphics[width=0.38\textwidth]{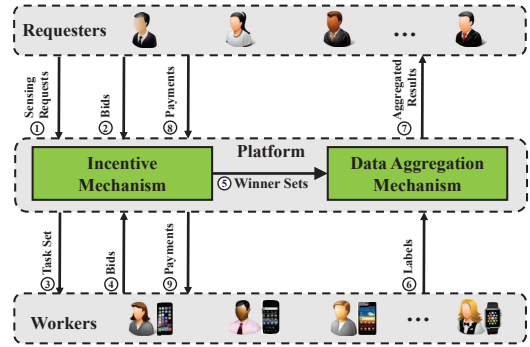}
\caption{Framework of CENTURION (where circled numbers represent the order of the events).}
\label{fig:sysframe}
\end{figure}

\begin{itemize}[leftmargin=*]\compresslist
\item \textbf{Incentive Mechanism.} Firstly, in the double auction-based incentive mechanism, each requester $r_j$ submits to the platform a sensing request containing the sensing task $\tau_j$ to be executed (step \circled{1}), and a bid $a_j$, the amount she is willing to pay if the task is executed (step \circled{2}). Then, the platform announces the set of sensing tasks $\mathcal{T}$ to the workers (step \circled{3}). After receiving the task set, every worker $w_i$ sends to the platform the set of tasks she wants to execute, denoted as $\Gamma_i\subseteq\mathcal{T}$, as well as a bid $b_i$, which is her bidding price for executing them (step \circled{4}). Based on received bids, the platform determines the set of winning requesters $\mathcal{S_R}$, the set of winning workers $\mathcal{S_W}$, as well as the payment $p_j^r$ charged from every winning requester $r_j$ and the payment $p_i^w$ paid to every winning worker $w_i$ (step \circled{5}). Note that losing requesters' tasks are not executed, and thus, they do not submit any payment. Similarly, losing workers do not receive any payment, as they do not execute any task. 
\item \textbf{Data Aggregation Mechanism.} Next, the platform collects the labels submitted by the winning workers (step \circled{6}), calculates the aggregated results, and sends them to the winning requesters (step \circled{7}). 
\item Finally, the platform charges $p_j^r$ from winning requester $r_j$ (step \circled{8}), and pays $p_i^w$ to winning worker $w_i$ (step \circled{9}).  
\end{itemize}

We denote the requesters' and workers' bid profile as $\mathbf{a}=(a_1,\cdots,a_M)$ and $\mathbf{b}=(b_1,\cdots,b_N)$, respectively. Moreover, the requesters' and workers' payment profile is denoted as $\mathbf{p}^r=(p_1^r,\cdots,p_M^r)$ and $\mathbf{p}^w=(p_1^w,\cdots,p_N^w)$, respectively.

\subsection{Reliability Level Model}
Before worker $w_i$ executes task $\tau_j$, her label about this task can be regarded as a random variable $L_{i,j}$. Then, we define the \textit{reliability level} of a worker in Definition \ref{def:reliabilitylevel}.

\begin{myDef}[Reliability Level]\label{def:reliabilitylevel}
A worker $w_i$'s reliability level $\theta_{i,j}$ about task $\tau_j$ is defined as the probability that she provides a correct label about this task, i.e.,
\begin{equation}
\begin{aligned}
\theta_{i,j}=\textnormal{\text{Pr}}[L_{i,j}=l_j]\in[0,1].
\end{aligned}
\end{equation}
Moreover, we denote the workers' reliability level matrix as $\boldsymbol{\uptheta}=[\theta_{i,j}]\in[0,1]^{N\times M}$.
\end{myDef}

We assume that the platform knows the reliability level matrix $\boldsymbol{\uptheta}$ \textit{a priori}, and maintains a historical record of it. In practice, the platform could obtain $\boldsymbol{\uptheta}$ through various approaches. For example, as, in many scenarios, workers tend to have similar reliability levels for similar tasks, the platform could assign to workers some tasks with known labels, and use workers' labels about these tasks to estimate their reliability levels for similar tasks as in \cite{DOlesonHCOMP11}. In cases where ground truth labels are not available, $\boldsymbol{\uptheta}$ can still be effectively inferred from workers' characteristics (e.g., the prices of a worker's sensors, a worker's experience and reputation for similar tasks) using the algorithms proposed in \cite{HLiWWW14}, or estimated using the labels previously submitted by workers about similar tasks by the methods in \cite{QLiSIGMOD14, CMengSenSys15}. 

\subsection{Auction Model}\label{sec:auctionmodel}

In this paper, we consider the scenario where both requesters and workers are \textit{strategic} and \textit{selfish} that aim to maximize their own utilities. Since CENTURION involves auctions among not only the workers, but also the requesters, we utilize the following \textit{double auction for \underline{M}ulti-r\underline{E}quester mobi\underline{L}e cr\underline{O}wd se\underline{N}sing (MELON double auction)}, formally defined in Definition \ref{def:auction}, as the incentive mechanism. 

\begin{myDef}[MELON Double Auction]\label{def:auction}
In a \textit{double auction for multi-requester mobile crowd sensing (MELON double auction)}, each requester $r_j$ obtains a value $v_j$, if her task $\tau_j$ is executed, and bids to the platform $a_j$, the amount she is willing to pay for the execution of her task. Each worker $w_i$ is interested in executing one subset of the tasks, denoted as $\Gamma_i\subseteq\mathcal{T}$, and bids to the platform $b_i$, her bidding price for executing these tasks. Her actual sensing cost for executing all tasks in $\Gamma_i$ is denoted as $c_i$. Both the requesters' values and workers' costs are unknown to the platform. 
\end{myDef}


Then, we define a requester's and worker's utility, as well as the platform's profit in Definition \ref{def:requesterutility}, \ref{def:workerutility}, and \ref{def:profit}. 
\begin{myDef}[Requester's Utility]\label{def:requesterutility}
A requester $r_j$'s utility is defined as
\begin{equation}
u_j^r=
\begin{cases}
\begin{aligned}
&v_j-p_j^r,&&\text{if~}r_j\in\mathcal{S_R}\\
&0,&&\text{otherwise}
\end{aligned}.
\end{cases}
\end{equation}
\end{myDef}

\begin{myDef}[Worker's Utility]\label{def:workerutility}
A worker $w_i$'s utility is defined as
\begin{equation}
u_i^w=
\begin{cases}
\begin{aligned}
&p_i^w-c_i,&&\text{if~}w_i\in\mathcal{S_W}\\
&0,&&\text{otherwise}
\end{aligned}.
\end{cases}
\end{equation}
\end{myDef}
\begin{myDef}[Platform's Profit]\label{def:profit}
The profit of the platform is defined as
\begin{equation}
\begin{aligned}
u_0=\sum_{j:r_j\in\mathcal{S_R}}p_j^r-\sum_{i:w_i\in\mathcal{S_W}}p_i^w.
\end{aligned}
\end{equation}
\end{myDef}

Based on Definition \ref{def:requesterutility}, \ref{def:workerutility}, and \ref{def:profit}, we define the social welfare of the MCS system in Definition \ref{def:socialwelfare}. 
\begin{myDef}[Social Welfare]\label{def:socialwelfare}
The social welfare of the MCS system is defined as
\begin{align}
u_{\textnormal{\text{social}}}&=u_0+\sum_{i:w_i\in\mathcal{W}}u_i^w+\sum_{j:r_j\in\mathcal{R}}u_j^r\notag\\
&=\sum_{j:r_j\in\mathcal{S_R}}v_j-\sum_{i:w_i\in\mathcal{S_W}}c_i.
\end{align}
\end{myDef}

Clearly, the social welfare is the sum of the platform's profit and all requesters' and workers' utilities. 

\subsection{Design Objectives}
In this paper, we aim to ensure that CENTURION bears the following advantageous properties. 

Since the requesters are strategic and selfish in our model, it is possible that any requester $r_j$ submits a bid $a_j$ that deviates from $v_j$ (i.e., her value for task $\tau_j$). Similarly, any worker $w_i$ might also submit a bid $b_i$ that differs from $c_i$ (i.e., her cost for executing all tasks in $\Gamma_i$). Thus, one of our objectives is to design a \textit{truthful} incentive mechanism defined in Definition \ref{def:truthfulness}.
\begin{myDef}[Truthfulness]\label{def:truthfulness}
A MELON double auction is truthful if and only if bidding $v_j$ and $c_i$ is the dominant strategy for each requester $r_j$ and worker $w_i$, i.e., bidding $v_j$ and $c_i$ maximizes, respectively, the utility of each requester $r_j$ and worker $w_i$, regardless of other requesters' and workers' bids. 
\end{myDef}

By definition \ref{def:truthfulness}, we aim to ensure that both requesters and workers bid truthfully to the platform. Apart from truthfulness, another desirable property that we aim to achieve is \textit{individual rationality} defined in Definition \ref{def:ir}. 
\begin{myDef}[Individual Rationality]\label{def:ir}
A MELON double auction is individual rational if and only if no requesters or workers receive negative utilities, i.e., we have $u_j^r\geq 0$, and $u_i^w\geq 0$, for every requester $r_j$ and worker $w_i$, respectively.
\end{myDef}

Individual rationality is a crucial property to stimulate the participation of both requesters and workers, because it ensures that the charge to a requester is no larger than her value, and a worker's sensing cost is also totally compensated. As mentioned in Section \ref{sec:sysoverview}, CENTURION aggregates workers' labels to ensure that the aggregated results have satisfactory accuracy, which is mathematically defined in Definition \ref{def:accuracy}.
\begin{myDef}[$\beta_j$-Accuracy]\label{def:accuracy}
A task $\tau_j$ is executed with $\beta_j$-accuracy if and only if $\textnormal{\text{Pr}}[\widehat{L}_j\not=l_j]\leq\beta_j$, where $\beta_j\in(0,1)$, and $\widehat{L}_j$ denotes the random variable representing the aggregated result for task $\tau_j$. 
\end{myDef}

By Definition \ref{def:accuracy}, $\beta_j$-accuracy ensures that the aggregated result equals to the true label with high probability. Note that, for every task $\tau_j$, $\beta_j$ is a parameter chosen by the platform, and a smaller $\beta_j$ implies a stronger requirement for the accuracy. 

In short, our objectives are to ensure that the proposed CENTURION framework provides satisfactory \textit{accuracy guarantee} for the aggregated results of all executed tasks, and incentivizes the participation of both requesters and workers in a \textit{truthful} and \textit{individual rational} manner. 
\section{Design Details}\label{sec:dd}

In this section, we present the design details of the incentive and data aggregation mechanism of CENTURION. 
\subsection{Data Aggregation Mechanism}
\subsubsection{Proposed Mechanism}\label{sec:mechanism}
~

Although the data aggregation mechanism follows the incentive mechanism in CENTURION's workflow, we introduce it first, as it affects the design of the incentive mechanism. 

In order to capture the effect of workers' diverse reliability on the calculation of the aggregated results, CENTURION adopts the following \textit{weighted aggregation method}. That is, the aggregated result $\widehat{l}_j$ for every executed task $\tau_j$ is calculated as 
\begin{equation}\label{eq:aggregation}
\begin{aligned}
\widehat{l}_j=\text{sign}\Bigg(\sum_{i:w_i\in\mathcal{S_W},\tau_j\in\Gamma_i}\lambda_{i,j}l_{i,j}\Bigg),
\end{aligned}
\end{equation}
where $\lambda_{i,j}>0$ is worker $w_i$'s weight on task $\tau_j$. Furthermore, the function $\text{sign}(x)$ equals to $+1$, if $x\geq 0$, and $-1$ otherwise.

Intuitively, higher weights should be assigned to workers who are more likely to submit correct labels, which makes the aggregated results closer to the labels provided by more reliable workers. In fact, many state-of-the-art literature \cite{QLiSIGMOD14, CMengSenSys15} utilize such weighted aggregation method to aggregate workers' data. As the weight $\lambda_{i,j}$'s highly affect the accuracy of the aggregated results, we propose, in the following Algorithm \ref{al:aggregation}, the data aggregation mechanism of CENTURION. 

\begin{algorithm}[h]
\small
\KwIn{$\boldsymbol{\uptheta}$, $\mathbf{l}$, $\boldsymbol{\Gamma}$, $\mathcal{S_R}$, $\mathcal{S_W}$\;}
\KwOut{$\big\{\widehat{l}_j|r_j\in\mathcal{S_R}\big\}$\;}
\ForEach{$j$ \text{s.t.} $r_j\in\mathcal{S_R}$}{\label{line:aggregationstart}
	$\widehat{l}_j\leftarrow\sum_{i:w_i\in\mathcal{S_W}, \tau_j\in\Gamma_i} \big(2\theta_{i,j}-1\big)l_{i,j}$\;\label{line:aggregation}
}
\Return $\big\{\widehat{l}_j|r_j\in\mathcal{S_R}\big\}$\;
\caption{Data Aggregation Mechanism}\label{al:aggregation}
\end{algorithm}

Algorithm \ref{al:aggregation} takes as inputs the reliability level matrix $\boldsymbol{\uptheta}$, the workers' label matrix $\mathbf{l}$, the profile of workers' interested task sets, denoted as $\boldsymbol{\Gamma}=(\Gamma_1,\cdots,\Gamma_N)$, the winning requester set $\mathcal{S_R}$, and the winning worker set $\mathcal{S_W}$. Note that a large $\theta_{i,j}$ indicates that a worker $w_i$ has a high reliability level for task $\tau_j$, and any worker $w_i$ with $\theta_{i,j}\leq 0.5$ will not be selected as a winner by the incentive mechanism. The aggregated result $\widehat{l}_j$ for each winning requester $r_j$'s task $\tau_j$ is calculated (line \ref{line:aggregationstart}-\ref{line:aggregation}) using Equation (\ref{eq:aggregation}) with the weight 
\begin{equation}\label{eq:weight}
\begin{aligned}
\lambda_{i,j}=2\theta_{i,j}-1,~\forall r_j\in\mathcal{S_R},~w_i\in\mathcal{S_W},~\tau_j\in\Gamma_i. 
\end{aligned}
\end{equation}

By Equation (\ref{eq:weight}), we have that $\lambda_{i,j}$, i.e., worker $w_i$'s weight for task $\tau_j$, increases with $\theta_{i,j}$, which conforms to our intuition that the higher the probability that worker $w_i$ provides a correct label about task $\tau_j$, the more her label $l_{i,j}$ should be counted in the calculation of the aggregated result about this task. We provide the formal analysis about the data aggregation mechanism in Section \ref{sec:analysisaggregation}. 

\subsubsection{Analysis}\label{sec:analysisaggregation}
~

In Theorem \ref{theo:accuracy}, we prove that the aggregated results calculated by Algorithm \ref{al:aggregation} has desirable accuracy guarantee.

\begin{myTheo}\label{theo:accuracy}
For each executed task $\tau_j$, the data aggregation mechanism given in Algorithm \ref{al:aggregation} minimizes the upper bound of the error probability of the aggregated result, i.e., $\textnormal{Pr}[\widehat{L}_j\not=l_j]$ (where $\widehat{L}_j$ is the random variable representing the aggregated result for task $\tau_j$ mentioned in Definition \ref{def:accuracy}), and satisfies that 
\begin{equation}
\begin{aligned}
\textnormal{Pr}[\widehat{L}_j\not=l_j]\leq\exp\Bigg(-\frac{\sum_{i:w_i\in\mathcal{S_W},\tau_j\in\Gamma_i}(2\theta_{i,j}-1)^2}{2}\Bigg).
\end{aligned}
\end{equation}
\end{myTheo}

\begin{proof}
We denote $X_{i,j}$ as the random variable for worker $w_i$'s weighted label about task $\tau_j$, i.e., $X_{i,j}=\lambda_{i,j} l_j$ with probability $\theta_{i,j}$, and $X_{i,j}=-\lambda_{i,j} l_j$ with probability $1-\theta_{i,j}$. Then, we define $X_j=\sum_{i:w_i\in\mathcal{S_W},\tau_j\in\Gamma_i} X_{i,j}$, and thus, $\mathbb{E}[X_j]=\sum_{i:w_i\in\mathcal{S_W},\tau_j\in\Gamma_i}\mathbb{E}[X_{i,j}]=\sum_{i:w_i\in\mathcal{S_W},\tau_j\in\Gamma_i}l_j\lambda_{i,j}(2\theta_{i,j}-1)$.

The error probability of the aggregated result can be calculated as $\text{Pr}[\widehat{L}_j\not=l_j]=\text{Pr}[X_j<0|l_j=1]\text{Pr}[l_j=1]+\text{Pr}[X_j\geq 0|l_j=-1]\text{Pr}[l_j=-1]$, and based on the Chernoff-Hoeffding bound, we have 
\begin{equation*}
\scalefont{0.84}
\begin{aligned}
\text{Pr}[X_j<0|l_j=1]&=\text{Pr}[\mathbb{E}[X_j]-X_j>\mathbb{E}[X_j]|l_j=1]\\
&\leq\exp\Bigg(-\frac{2(\mathbb{E}[X_j|l_j=1])^2}{\sum_{i:w_i\in\mathcal{S_W},\tau_j\in\Gamma_i}(2\lambda_{i,j})^2}\Bigg)\\
&=\exp\Bigg(-\frac{\big(\sum_{i:w_i\in\mathcal{S_W},\tau_j\in\Gamma_i}\lambda_{i,j}(2\theta_{i,j}-1)\big)^2}{2\sum_{i:w_i\in\mathcal{S_W},\tau_j\in\Gamma_i}\lambda_{i,j}^2}\Bigg).
\end{aligned}
\end{equation*}

Then, we define the vector $\boldsymbol{\uplambda}_j=[\lambda_{i,j}]$ for every executed task $\tau_j$, which contains every $\lambda_{i,j}$ such that $w_i\in\mathcal{S_W}$, and $\tau_j\in\Gamma_i$. Therefore, minimizing the upper bound of $\text{Pr}[X_j<0|l_j=1]$ is equivalent to finding the vector $\boldsymbol{\uplambda}_j$ that maximizes the function $f(\boldsymbol{\uplambda}_j)$ defined as 

\begin{equation*}\label{eq:flambda}
\begin{aligned}
f(\boldsymbol{\uplambda}_j)=\frac{\big(\sum_{i:w_i\in\mathcal{S_W},\tau_j\in\Gamma_i}\lambda_{i,j}(2\theta_{i,j}-1)\big)^2}{\sum_{i:w_i\in\mathcal{S_W},\tau_j\in\Gamma_i}\lambda_{i,j}^2}.
\end{aligned}
\end{equation*}

Based on the Cauchy-Schwarz inequality, we have 

\begin{equation*}
\scalefont{0.91}
\begin{aligned}
f(\boldsymbol{\uplambda}_j)&\leq\frac{\big(\sum_{i:w_i\in\mathcal{S_W},\tau_j\in\Gamma_i}\lambda_{i,j}^2\big)\big(\sum_{i:w_i\in\mathcal{S_W},\tau_j\in\Gamma_i}(2\theta_{i,j}-1)^2\big)}{\sum_{i:w_i\in\mathcal{S_W},\tau_j\in\Gamma_i}\lambda_{i,j}^2}\\
&=\sum_{i:w_i\in\mathcal{S_W},\tau_j\in\Gamma_i}(2\theta_{i,j}-1)^2,
\end{aligned}
\end{equation*}
and equality is achieved if and only if $\lambda_{i,j}\propto 2\theta_{i,j}-1$. Thus, 
\begin{equation}\label{eq:conprob1}
\scalefont{0.84}
\begin{aligned}
\text{Pr}[X_j<0|l_j=1]\leq\exp\Bigg(-\frac{\sum_{i:w_i\in\mathcal{S_W},\tau_j\in\Gamma_i}(2\theta_{i,j}-1)^2}{2}\Bigg).
\end{aligned}
\end{equation}

Similarly, from the Chernoff-Hoeffding bound, we have 
\begin{equation*}
\scalefont{0.78}
\begin{aligned}
\text{Pr}[X_j\geq0|l_j=-1]\leq\exp\Bigg(-\frac{\big(\sum_{i:w_i\in\mathcal{S_W},\tau_j\in\Gamma_i}\lambda_{i,j}(2\theta_{i,j}-1)\big)^2}{2\sum_{i:w_i\in\mathcal{S_W},\tau_j\in\Gamma_i}\lambda_{i,j}^2}\Bigg).
\end{aligned}
\end{equation*}

The upper bound of $\text{Pr}[X_j>0|l_j=-1]$ is also minimized if and only if $\lambda_{i,j}\propto 2\theta_{i,j}-1$ based on the Cauchy-Schwarz inequality, and we have 
\begin{equation}\label{eq:conprob2}
\scalefont{0.8}
\begin{aligned}
\text{Pr}[X_j\geq 0|l_j=-1]\leq\exp\Bigg(-\frac{\sum_{i:w_i\in\mathcal{S_W},\tau_j\in\Gamma_i}(2\theta_{i,j}-1)^2}{2}\Bigg).
\end{aligned}
\end{equation}

From Inequality (\ref{eq:conprob1}) and (\ref{eq:conprob2}), we have that when $\lambda_{i,j}=2\theta_{i,j}-1$, the upper bound of $\textnormal{Pr}[\widehat{L}_j\not=l_j]$ is minimized, and 
\begin{equation*}
\begin{aligned}
\textnormal{Pr}[\widehat{L}_j\not=l_j]\leq\exp\Bigg(-\frac{\sum_{i:w_i\in\mathcal{S_W},\tau_j\in\Gamma_i}(2\theta_{i,j}-1)^2}{2}\Bigg),
\end{aligned}
\end{equation*}
which exactly proves Theorem \ref{theo:accuracy}. 
\end{proof}

By Theorem \ref{theo:accuracy}, we have that the data aggregation mechanism proposed in Algorithm \ref{al:aggregation} upper bounds the error probability $\textnormal{Pr}[\widehat{L}_j\not=l_j]$ by $\exp\big(-\frac{1}{2}\sum_{i:w_i\in\mathcal{S_W},\tau_j\in\Gamma_i}(2\theta_{i,j}-1)^2\big)$, which in fact is the minimum upper bound of this probability. Next, we derive Corollary \ref{cor:accurcay}, which is directly utilized in our design of the incentive mechanism in Section \ref{sec:incentivemechanism}. 

\begin{myCor}\label{cor:accurcay}
For every executed task $\tau_j$, the data aggregation mechanism proposed in Algorithm \ref{al:aggregation} satisfies that if 
\begin{equation}\label{eq:errConstraintinCor}
\begin{aligned}
\sum_{i:w_i\in\mathcal{S_W},\tau_j\in\Gamma_i}(2\theta_{i,j}-1)^2\geq 2\ln\bigg(\frac{1}{\beta_j}\bigg),
\end{aligned}
\end{equation}
then $\textnormal{Pr}[\widehat{L}_j\not=l_j]\leq\beta_j$, i.e., $\beta_j$-accuracy is satisfied for this task $\tau_j$, where $\beta_j\in(0,1)$ is a platform chosen parameter. Moreover, we define $\boldsymbol{\upbeta}$ as the vector $(\beta_1,\cdots, \beta_M)$. 
\end{myCor}

\begin{proof}
By setting the upper bound of $\textnormal{Pr}[\widehat{L}_j\not=l_j]$ given in Theorem \ref{theo:accuracy} to be no greater than $\beta_j\in(0,1)$, we have 
\begin{equation*}
\begin{aligned}
\exp\Bigg(-\frac{\sum_{i:w_i\in\mathcal{S_W},\tau_j\in\Gamma_i}(2\theta_{i,j}-1)^2}{2}\Bigg)\leq\beta_j,
\end{aligned}
\end{equation*}
which is equivalent to 
\begin{equation}\label{eq:errConstraintinProof}
\begin{aligned}
\sum_{i:w_i\in\mathcal{S_W},\tau_j\in\Gamma_i}(2\theta_{i,j}-1)^2\geq 2\ln\bigg(\frac{1}{\beta_j}\bigg). 
\end{aligned}
\end{equation}
Hence, together with Theorem \ref{theo:accuracy}, we have that Inequality (\ref{eq:errConstraintinProof}) indicates that $\textnormal{Pr}[\widehat{L}_j\not=l_j]\leq\beta_j$. 
\end{proof}

Corollary \ref{cor:accurcay} gives us a sufficient condition, represented by Inequality (\ref{eq:errConstraintinCor}), that the set of winning workers $\mathcal{S_W}$ selected by the incentive mechanism (proposed in Section \ref{sec:incentivemechanism}) should satisfy so as to achieve $\beta_j$-accuracy for each executed task $\tau_j$. 

\subsection{Incentive Mechanism}\label{sec:incentivemechanism}
Now, we introduce the design details of CENTURION's incentive mechanism, including its mathematical formulation, the hardness proof of the formulated integer program, the proposed mechanism, as well as the corresponding analysis.  
\subsubsection{Mathematical Formulation}
~

As mentioned in Section \ref{sec:auctionmodel}, CENTURION's incentive mechanism is based on the MELON  double auction defined in Definition \ref{def:auction}. In this paper, we aim to design a MELON double auction that \textit{maximizes the social welfare}, while guaranteeing \textit{satisfactory data aggregation accuracy}. The formal mathematical formulation of its winner selection problem is provided in the following \textit{MELON double auction social welfare maximization (MELON-SWM) problem.}
~\\

\textbf{MELON-SWM Problem:}
\begin{small}
\begin{align}
\max&\sum_{j:\tau_j\in\mathcal{T}}a_jy_j-\sum_{i:w_i\in\mathcal{W}}b_ix_i\\
\text{s.t.}& \sum_{i:w_i\in\mathcal{W},\tau_j\in\Gamma_i}(2\theta_{i,j}-1)^2 x_i\geq 2\ln\bigg(\frac{1}{\beta_j}\bigg)y_j,~\forall\tau_j\in\mathcal{T}\label{eq:constraint}\\
&x_i, y_j\in\{0,1\},~\forall w_i\in\mathcal{W},~\tau_j\in\mathcal{T}
\end{align}
\end{small}
\vspace{-0.4cm}

\textbf{Constants.} The MELON-SWM problem takes as inputs the task set $\mathcal{T}$, the worker set $\mathcal{W}$, the requesters' and workers' bid profile $\mathbf{a}$ and $\mathbf{b}$, the profile of workers' interested task sets $\boldsymbol{\Gamma}$, the workers' reliability level matrix $\boldsymbol{\uptheta}$, and the $\boldsymbol{\upbeta}$ vector. 

\textbf{Variables.} On one hand, the MELON-SWM problem has a vector of $M$ binary variables, denoted as $\mathbf{y}=(y_1,\cdots,y_M)$. Any $y_j=1$ indicates that task $\tau_j$ will be executed, and thus, requester $r_j$ is a winning requester (i.e., $r_j\in\mathcal{S_R}$), whereas $y_j=0$ means $r_j\not\in\mathcal{S_R}$. On the other hand, the problem has another vector of $N$ binary variables, denoted as $\mathbf{x}=(x_1,\cdots,x_N)$, where $x_i=1$ indicates that worker $w_i$ is a winning worker (i.e., $w_i\in\mathcal{S_W}$), and $x_i=0$ means $w_i\not\in\mathcal{S_W}$. 

\textbf{Objective function.} The objective function satisfies that $\sum_{j:\tau_j\in\mathcal{T}}a_jy_j-\sum_{i:w_i\in\mathcal{W}}b_ix_i=\sum_{j:r_j\in\mathcal{S_R}}a_j-\sum_{i:w_i\in\mathcal{S_W}}b_i$, which is exactly the social welfare defined in Definition \ref{def:socialwelfare} based on the requesters' and workers' bids. 

\textbf{Constraints.} For each task $\tau_j$, Constraint (\ref{eq:constraint}) naturally holds, if $y_j=0$. When $y_j=1$, it is equivalent to Inequality (\ref{eq:errConstraintinCor}) given in Corollary \ref{cor:accurcay}, which specifies the condition that the set of selected winning workers $\mathcal{S_W}$ should satisfy in order to guarantee $\beta_j$-accuracy for task $\tau_j$. To simplify the presentation, we introduce the following notations, namely $q_{i,j}=(2\theta_{i,j}-1)^2$, $\mathbf{q}=[q_{i,j}]\in[0,1]^{N\times M}$, $Q_j=2\ln\big(\frac{1}{\beta_j}\big)$, and $\mathbf{Q}=[Q_j]\in[0,+\infty)^{M\times 1}$. Thus, Constraint (\ref{eq:constraint}) can be simplified as 
\begin{equation}\label{eq:constraintsimple}
\begin{aligned}
\sum_{i:w_i\in\mathcal{W},\tau_j\in\Gamma_i}q_{i,j} x_i\geq Q_j y_j,~\forall\tau_j\in\mathcal{T}.
\end{aligned}
\end{equation}
Besides, we say a task $\tau_j$ is \textit{covered} by a solution, if $y_j=1$. 
\subsubsection{Hardness Proof}
~

We prove the NP-hardness of the MELON-SWM problem by performing a polynomial-time reduction from the 3SAT(5) problem which is formally defined in Definition \ref{def:3SAT}.

\begin{myDef}[3SAT(5) Problem]\label{def:3SAT}
In a 3SAT(5) problem, we are given a set $\mathcal{O}=\{z_1,\cdots,z_n\}$ of $n$ Boolean variables, and a collection $C_1,\cdots,C_m$ of $m$ clauses. Each clause is an OR of exactly three literals, and every literal is either a variable of $\mathcal{O}$ or its negation. Moreover, every variable participates in exactly 5 clauses. Therefore, $m=\frac{5n}{3}$. Given some constant $0<\epsilon<1$, a 3SAT(5) instance $\varphi$ is a Yes-Instance if there is an assignment to the variables of $\mathcal{O}$ satisfying all clauses, whereas it is a No-Instance (with respect to $\epsilon$), if every assignment to the variables satisfies at most $(1-\epsilon)m$ clauses. An algorithm $\mathcal{A}$ distinguishes between the Yes- and No-instances of the problem, if, given a Yes-Instance, it returns a ``YES'' answer, and given a No-Instance it returns a ``NO'' answer.
\end{myDef}

Regarding the hardness of the 3SAT(5) problem, we introduce without proof the following well-known Lemma \ref{theo:SATHard}, which is a consequence of the PCP theorem \cite{SAroraJACM98}.

\begin{myLemma}\label{theo:SATHard}
There is some constant $0<\epsilon<1$, such that distinguishing between the Yes- and No-instances of the 3SAT(5) problem, defined with respect to $\epsilon$, is NP-complete. 
\end{myLemma}

Next, we introduce Theorem \ref{theo:reduction} and \ref{theo:welfarematching} that will be utilized to prove the NP-hardness of the MELON-SWM problem. 

\begin{myTheo}\label{theo:reduction}
Any 3SAT(5) instance is polynomial-time reducible to an instance of the MELON-SWM problem. 
\end{myTheo}
\begin{proof}
The reduction goes as follows. Assume there is a 3SAT(5) instance $\varphi$ on $n$ variables and $m$ clauses. We define 3 parameters: $X=\frac{\epsilon m}{100}$ ($0<\epsilon<1$), $Y=mnX$, and $Z=mnY$. The exact values of $Y$ and $Z$ are not important. We just need to ensure $Z\gg Y\gg X$. We construct an instance of the MELON-SWM problem corresponding to $\varphi$, by defining the task set $\mathcal{T}$, and the profile of workers' interested task sets $\boldsymbol{\Gamma}$. 

Out of the 8 possible assignments to the variables of some clause $C_k\in\varphi$, exactly one does not satisfy $C_k$. Let $A_k$ be the set of the remaining 7 assignments. We define a set of tasks $\Gamma(C_k,\alpha)$ for each clause $C_k$ and assignment $\alpha\in A_k$, let $\boldsymbol{\Gamma}=[\Gamma(C_k, \alpha)]$ for each clause $C_k\in\varphi$ and assignment $\alpha\in A_k$, set the $q_{i,j}$ value of each worker $w_i$ and task $\tau_j\in\Gamma_i$ as $q_{i,j}=1$, and set her bid as $b_i=3+Y+Z$. We also create a dummy worker $w_0$, with $q_0=1$, $b_0=0$, and $\Gamma_0$ being her interested task set. We start with all set $\Gamma(C_k, \alpha)$'s being empty, gradually define the tasks, and specify which sets they belong to. The task set $\mathcal{T}$ consists of 4 subsets.

\begin{itemize}[leftmargin=*]\compresslist
\item The 1st subset $E_1$ contains a task $\tau(z_l, \gamma)$ for each variable $z_l\in\mathcal{O}$ and assignment $\gamma\in\{T,F\}$ to this variable. $\tau(z_l, \gamma)$ belongs to each set $\Gamma(C_k , \alpha)$, such that $z_l$ participates in $C_k$, and the assignment $\alpha$ to the variables of $C_k$ gives assignment $\gamma$ to $z_l$. The $Q_j$ value of the task $\tau_j$ corresponding to $\tau(z_l, \gamma)$ is set as $5-\text{the number of the clauses containing~}z_l$, and the value $v_j$ of this task is set as $5$.
\item The 2nd subset $E_2$ contains $m$ tasks $\tau_1,\cdots,\tau_m$. Each $\tau_k\in E_2$ belongs to all sets corresponding to $C_k$ and $C_{k+1}$, i.e., $\tau_k$ belongs to all sets $\{\Gamma(C_k, \alpha)|\alpha\in A_k\}\cup\{\Gamma(C_{k+1},\alpha')|\alpha'\in A_{k+1}\}$ with the subscripts being modulo $m$. The $Q_k$ value of each such $\tau_k$ is set as $2$, and its value $v_k$ is set as $Y$. 
\item The 3rd subset $E_3$ contains a task $\tau(C_k)$ for each clause $C_k$, and $\tau(C_k)$ belongs to set $\Gamma(C_k,\alpha)$ for each $\alpha\in A_k$. The $Q_j$ value of the task $\tau_j$ corresponding to $\tau(C_k)$ is set as 1, and its value $v_j$ is set as $Z$. 
\item The 4th subset $E_4$ contains a single task $\tau^*$, whose $Q_j$ value is set as $1$ and value $v_j$ is set as $X$. The task $\tau^*$ only belongs to set $\Gamma_0$.
\end{itemize}

This finishes the description of the reduction. Clearly, given a 3SAT(5) instance $\varphi$, we can construct an instance of the MELON-SWM problem in time polynomial in $n$. 
\end{proof}
We now analyze the optimal social welfare for an instance of the MELON-SWM problem that corresponds to a 3SAT(5) instance $\varphi$, when $\varphi$ is a Yes- or No-Instance. Note that the following analysis uses the same reduction as in Theorem \ref{theo:reduction}. 

\begin{myTheo}\label{theo:welfarematching}
If the 3SAT(5) instance $\varphi$ is a Yes-Instance, then there is a solution to the resulting instance of the MELON-SWM problem whose social welfare is $X$. If $\varphi$ is a No-Instance, then any solution has social welfare at most $0$. 
\end{myTheo}
\begin{proof}
Let $\varphi$ be a Yes-Instance, and $A$ be an assignment to the variables satisfying all clauses. We construct a solution $\mathcal{S}'$ to the MELON-SWM problem. Firstly, we add $\Gamma_0$ to $\mathcal{S}'$. Next, for each clause $C_k$, we add to $\mathcal{S}'$ the unique set $\Gamma(C_k,\alpha)$, where $\alpha$ is the assignment consistent with $A$. Then $|\mathcal{S}'|=m$, and the total cost of all sets is $(Y+Z+3)m$. We now analyze the number of tasks covered by $\mathcal{S}'$, and their values. Clearly, $\tau^*$ is covered by $\mathcal{S}'$, and it contributes $X$ to the solution value. 
\begin{itemize}[leftmargin=*]\compresslist
\item For each clause $C_k\in\varphi$, the unique task $\tau(C_k)\in E_3$ is covered. Thus, all tasks in $E_3$ are covered, and overall they contribute value $mZ$ to the solution. 
\item Consider some $\tau_k\in E_2$. $\mathcal{S}'$ contains one set corresponding to $C_k$ and $C_{k+1}$, respectively. Since $\tau_k$ belongs to both these sets, and its $Q_k$ is 2, it is covered. Thus, all tasks in $E_2$ are covered, and they contribute value $mY$ to the solution. 
\item Consider some variable $z_k\in\mathcal{O}$, and let $\gamma_k\in\{T,F\}$ be the assignment to $z_k$ under $A$. If $C_k$ is any clause containing $z_k$, and $\Gamma(C_k, \alpha)$ is the set that belongs to $\mathcal{S}'$, then $\alpha$ gives the assignment $\gamma_k$ to $z_k$. Thus, for all five clauses containing $z_k$, the corresponding sets chosen to $\mathcal{S}'$ contain $\tau(z_k,\gamma_k)$, and this task is covered. So the total number of tasks of $E_1$ covered by $\mathcal{S}'$ is $n$. Each such task contributes value 5, and the total value contributed by the tasks in $E_1$ is $5n=3m$. 
\end{itemize}

Therefore, the overall social welfare of this solution is $X+mZ+mY+3m-(Z+Y+3)m=X$. 

Assume now that $\varphi$ is a No-Instance, and let $\mathcal{S}'$ be any solution with positive social welfare. We can assume that $\Gamma_0\in\mathcal{S}'$, and task $\tau^*$ is covered by $\mathcal{S}'$. We then introduce the following observations, whose proofs are provided in the appendices. 

\begin{myObs}\label{theo:obs1}
For every clause $C_k$ of $\varphi$, at most one of the sets $\{\Gamma(C_k,\alpha)|\alpha\in A_k\}$ belongs to $\mathcal{S}'$, and $|\mathcal{S}'|=m$.
\end{myObs}
\begin{myObs}\label{theo:obs3}
For every variable $z_k\in\mathcal{O}$, at most one of the two tasks $\tau(z_k, T)$ and $\tau(z_k, F)$ is covered by $\mathcal{S}'$.
\end{myObs}

We say that a variable $z_k\in\mathcal{O}$ is bad if neither $\tau(z_k, T)$ nor $\tau(z_k, F)$ is covered by $\mathcal{S}'$; otherwise it is good. We next show that only a small number of the variables are bad.
\begin{myObs}\label{theo:obs4}
There are at most $\frac{\epsilon n}{100}$ bad variables.
\end{myObs}

Then, we construct the following assignment to the variables of $\mathcal{O}$. If variable $z_k\in\mathcal{O}$ is good, then there is a unique value $\gamma_k\in\{T,F\}$, such that task $\tau(z_k, \gamma_k)$ is covered by $\mathcal{S}'$. We then assign $z_k$ the value $\gamma_k$. If $z_k$ is bad, we assign it any value arbitrarily. We now claim that the above assignment satisfies more than $(1-\epsilon)m$ clauses. We say that a clause is bad if it contains a bad variable, and it is good otherwise. Since there are at most $\frac{\epsilon n}{100}$ bad variables, and each variable participates in 5 clauses, the number of bad clauses is at most $\frac{\epsilon n}{20}\leq\frac{3\epsilon m}{100}$. So there are more than $(1-\epsilon)m$ good clauses. Let $C_l$ be a good clause, and $\Gamma(C_l,\alpha)$ be the set corresponding to $C_l$ that belongs to $\mathcal{S}'$. Then $\alpha$ is an assignment to the variables of $C_l$ that satisfies $C_l$, and each variable participating in $C_l$ was assigned a value consistent with $\alpha$. So clause $C_l$ is satisfied. 

To conclude, we have assumed that $\varphi$ is a No-Instance, and showed that, if the MELON-SWM problem has a solution with non-negative social welfare, there is an assignment to the variables of $\varphi$ satisfying more than $(1-\epsilon)m$ of its clauses, which is impossible for a No-Instance. Therefore, if $\varphi$ is a No-Instance, every solution has social welfare at most $0$.
\end{proof}

Next, we describe Theorem \ref{theo:NPHard} that states the \textit{NP-hardness} and \textit{inapproximability} of the MELON-SWM problem.

\begin{myTheo}\label{theo:NPHard}
The MELON-SWM problem is NP-hard, and for any factor $\phi$, there is no efficient $\phi$-approximation algorithm to the MELON-SWM problem. 
\end{myTheo}

\begin{proof}
Based on Theorem \ref{theo:reduction}, there exists a reduction from any 3SAT(5) problem instance $\varphi$ to an instance $\mathcal{I}(\varphi)$ of the MELON-SWM problem. From Theorem \ref{theo:welfarematching}, we have that the optimal solution to $\mathcal{I}(\varphi)$ also gives a solution to $\varphi$. That is, if the optimal social welfare of $\mathcal{I}(\varphi)$ is positive, then $\varphi$ is a Yes-Instance; otherwise, $\varphi$ is a No-Instance. Together with Lemma \ref{theo:SATHard} stating the NP-completeness of the 3SAT(5) problem, we conclude that the MELON-SWM problem is NP-hard. 

In fact, Theorem \ref{theo:reduction} and \ref{theo:welfarematching} give an inapproximability result about the MELON-SWM, as well. Suppose there is an efficient factor-$\phi$ approximation algorithm $\mathcal{A}$ for the MELON-SWM problem. We can use it to distinguish Yes- and No-instances of the 3SAT(5) problem on $n\gg\phi$ variables. If $\varphi$ is a Yes-Instance, then the algorithm has to return a solution with positive social welfare for $\mathcal{I}(\varphi)$, and if $\varphi$ is a No-Instance, then any solution has social welfare at most 0. So algorithm $\mathcal{A}$ distinguishes the Yes- and the No-instances of 3SAT(5), contradicting Lemma \ref{theo:SATHard}.
\end{proof}

\subsubsection{Proposed Mechanism}
~

Theorem \ref{theo:NPHard} not only shows the NP-hardness of the MELON-SWM problem, but also indicates that there is no efficient algorithm with a guaranteed approximation ratio for it. Therefore, we relax the requirement of provable approximation ratio, and propose the following MELON double auction that aims to ensure \textit{non-negative social welfare}, instead. Its winner selection algorithm is given in the following Algorithm \ref{al:wd}.

\begin{algorithm}\label{al:wd}
\SetAlgoNlRelativeSize{0}
\SetNlSkip{0.3em}
\small
\KwIn{$\mathcal{T}$, $\mathcal{R}$, $\mathcal{W}$, $\boldsymbol{\Gamma}$, $\mathbf{a}$, $\mathbf{b}$, $\mathbf{q}$, $\mathbf{Q}$\;}
\KwOut{$\mathcal{S_R}$, $\mathcal{S_W}$, $\mathcal{C}$\;}
\tcp{Initialization}
$\mathcal{S_R}\leftarrow\emptyset$, $\mathcal{S_W}\leftarrow \emptyset$\;\label{line:iniwd}
\tcp{Find a feasible cover}
$\mathcal{C}\leftarrow\texttt{FC}(\mathcal{T}, \boldsymbol{\Gamma}, \mathbf{q}, \mathbf{Q})$\;\label{line:feasiblecover}
\ForEach{$j$ \text{s.t.} $\tau_j\in\mathcal{T}$}{\label{line:Cjforstart}
	$\mathcal{C}_j\leftarrow\{w_i|w_i\in\mathcal{C}, \tau_j\in\Gamma_i\}$\;\label{line:Cjforend}
}
\tcp{Main loop}
\While{$\max_{j:r_j\in\mathcal{R}}\big(a_j-\sum_{i:w_i\in\mathcal{C}_j}b_i\big)\geq 0$}{\label{line:mainloopwinnerstart}
	$j^*\leftarrow\arg\max_{j:r_j\in\mathcal{R}}\big(a_j-\sum_{i:w_i\in\mathcal{C}_j}b_i\big)$\;\label{line:maxindex}
	$\mathcal{S_R}\leftarrow\mathcal{S_R}\cup\{r_{j^*}\}$\;\label{line:incrstar}
	$\mathcal{R}\leftarrow\mathcal{R}\setminus\{r_{j^*}\}$\;\label{line:excrstar}
	$\mathcal{S_W}\leftarrow\mathcal{S_W}\cup\mathcal{C}_{j^*}$\;\label{line:incCjstar}
	\ForEach{$j$ \text{s.t.} $r_i\in\mathcal{R}$}{\label{line:innerforwinner}
		$\mathcal{C}_j\leftarrow\mathcal{C}_j\setminus\mathcal{C}_{j^*}$\;\label{line:mainloopwinnerend}
	}
}
\Return$\mathcal{S_R}, \mathcal{S_W}$\;\label{line:returnwinner}
\caption{MELON Double Auction Winner Selection}
\end{algorithm}

Algorithm \ref{al:wd} takes as inputs the task set $\mathcal{T}$, the requester set $\mathcal{R}$, the worker set $\mathcal{W}$, the profile of workers' interested task sets $\boldsymbol{\Gamma}$, the requesters' and workers' bid profile $\mathbf{a}$ and $\mathbf{b}$, the $\mathbf{q}$ matrix, as well as the $\mathbf{Q}$ vector. Firstly, it initializes the winning requester and worker set as $\emptyset$ (line \ref{line:iniwd}). Then, it calculates a \textit{feasible cover}, denoted by $\mathcal{C}$, containing the set of workers that make Constraint (\ref{eq:constraintsimple}) feasible for each task $\tau_j$ given that each $y_j=1$, by calling another algorithm $\texttt{FC}$ which takes the task set $\mathcal{T}$, the profile of workers' interested task sets $\boldsymbol{\Gamma}$, the $\mathbf{q}$ matrix, and the $\mathbf{Q}$ vector as inputs (line \ref{line:feasiblecover}). Algorithm $\texttt{FC}$ can be easily implemented in time polynomial in $M$ and $N$. For example, $\texttt{FC}$ could greedily select each worker $w_i$ into the feasible cover in a decreasing order of the value $\sum_{j:\tau_j\in\Gamma_i}q_{i,j}$ until all constraints are satisfied. The computational complexity of such $\texttt{FC}$ is $O(N)$. We assume that $\texttt{FC}$ adopts such a greedy approach in the rest of this paper. Note that the specific choice of $\texttt{FC}$ is not important, as long as it returns a feasible cover in polynomial time. Next, for each task $\tau_j$, Algorithm \ref{al:wd} chooses from the feasible cover the set of workers $\mathcal{C}_j$ whose interested task sets contain this task (line \ref{line:Cjforstart}-\ref{line:Cjforend}). 

Based on $\mathcal{C}$, the main loop (line \ref{line:mainloopwinnerstart}-\ref{line:mainloopwinnerend}) of the algorithm selects the set of winning requesters and workers that give non-negative social welfare. It executes until $\max_{j:r_j\in\mathcal{R}}\big(a_j-\sum_{i:w_i\in\mathcal{C}_j}b_i\big)$, the \textit{maximum marginal social welfare} of including a new requester $r_j$ and the set of workers $\mathcal{C}_j$ into, respectively, the winning requester and worker set, becomes negative (line \ref{line:mainloopwinnerstart}). In each iteration of the main loop, the Algorithm finds first the index $j^*$ of the requester $r_{j^*}$ that provides the \textit{maximum marginal social welfare} (line \ref{line:maxindex}). Next, it includes $r_{j^*}$ into the winning requester set $\mathcal{S_R}$ (line \ref{line:incrstar}), removes $r_{j^*}$ from the requester set $\mathcal{R}$ (line \ref{line:excrstar}), and includes all workers in $\mathcal{C}_{j^*}$ into the winning worker set $\mathcal{S_W}$ (line \ref{line:incCjstar}). The last step of the main loop is to remove all workers in $\mathcal{C}_{j^*}$ from $\mathcal{C}_j$ for each task $\tau_j$ (line \ref{line:innerforwinner}). Finally, Algorithm \ref{al:wd} returns the winning requester and worker set $\mathcal{S_R}$ and $\mathcal{S_W}$ (line \ref{line:returnwinner}). 

Next, we present the pricing algorithm of the MELON double auction in Algorithm \ref{al:pricing}. 

\begin{algorithm}\label{al:pricing}
\SetAlgoNlRelativeSize{0}
\SetNlSkip{0.3em}
\small
\KwIn{$\mathcal{T}$, $\mathcal{R}$, $\mathcal{W}$, $\boldsymbol{\Gamma}$, $\mathbf{a}$, $\mathbf{b}$, $\mathbf{q}$, $\mathbf{Q}$, $\mathcal{S_R}$, $\mathcal{S_W}$\;}
\KwOut{$\mathbf{p}^r$, $\mathbf{p}^w$\;}
\tcp{Initialization}
$\mathbf{p}^r\leftarrow\mathbf{0}$, $\mathbf{p}^w\leftarrow\mathbf{0}$\;\label{line:inipricing}
\tcp{Pricing for winning requesters}
\ForEach{$j$ \text{s.t.} $r_j\in\mathcal{S_R}$}{\label{line:looprpstart}
	run Algorithm \ref{al:wd} on $\mathcal{R}\setminus \{r_j\}$ and $\mathcal{W}$\;\label{line:wdinrp}
	$\mathcal{S'_R}\leftarrow$winning requester set when line \ref{line:wdinrp} stops\;\label{line:imsets}
	\ForEach{$k$ \text{s.t.} $r_k\in\mathcal{S'_R}$}{\label{line:innerforrpstart}
		$p_j^r\leftarrow\min\big\{p_j^r, \sum_{w_i\in\mathcal{C}'_j}b_i+a_k-\sum_{w_i\in\mathcal{C}'_k}b_i\big\}$\;\label{line:innerforrpend}
	}
	\If{$\mathcal{C}'_j=\emptyset$}{\label{line:ifempty}
		$p_j^r\leftarrow\min\{p_j^r, 0\}$\;\label{line:looprpend}
	}
}
\tcp{Pricing for winning workers}
\ForEach{$i$ \text{s.t.} $w_i\in\mathcal{S_W}$}{\label{line:loopwpstart}
	run Algorithm \ref{al:wd} on $\mathcal{R}$ and $\mathcal{W}\setminus \{w_i\}$\;\label{line:wdinwp}
	$\mathcal{S'_R}\leftarrow$winning requester set when line \ref{line:wdinwp} stops\;\label{line:imsetswp}
	\ForEach{$k$ \text{s.t.} $w_i\in\mathcal{C}'_k$ and $r_k\in\mathcal{S'_R}$}{\label{line:forwpstart}
		sort requesters according to the decreasing order of $a_j-\sum_{i:w_i\in\mathcal{C}'_j}b_i$\;\label{line:sort}
		$f\leftarrow$index of the first requester with $w_i\not\in\mathcal{C}'_f$\;\label{line:firstindex}
		\eIf{$r_f\in\mathcal{S'_R}$}{\label{line:ifinSprimeR}
			$p_i^w\leftarrow\max\big\{p_i^w, a_k-\sum_{w_h\in\mathcal{C}'_k}b_h-\big(a_f-\sum_{w_h\in\mathcal{C}'_f}b_h\big)\big\}$\;
		}{
			$p_i^w\leftarrow\max\big\{p_i^w, a_k-\sum_{w_h\in\mathcal{C}'_k}b_h\big\}$\;\label{line:loopwpend}
		}
	}
}
\Return$\mathbf{p}^r$, $\mathbf{p}^w$\;\label{line:returnpricing}
\caption{MELON Double Auction Pricing}
\end{algorithm}

Apart from the same inputs to Algorithm \ref{al:wd}, Algorithm \ref{al:pricing} also takes as inputs the winning requester and worker set $\mathcal{S_R}$ and $\mathcal{S_W}$, outputted by Algorithm \ref{al:wd}. Firstly, Algorithm \ref{al:pricing} initializes the requesters' and workers' payment profile as zero vectors (line \ref{line:inipricing}). Then, it calculates the payment $p_j^r$ charged from each winning requester (line \ref{line:looprpstart}-\ref{line:looprpend}). For each $r_j\in\mathcal{S_R}$, Algorithm \ref{al:wd} is executed on the worker set $\mathcal{W}$ and requester set $\mathcal{R}$ except requester $r_j$ (line \ref{line:wdinrp}). Next, it sets $\mathcal{S'_R}$ as the winning requester set when line \ref{line:wdinrp} stops (line \ref{line:imsets}). For each $r_k\in\mathcal{S'_R}$, Algorithm \ref{al:pricing} finds the minimum bid $a_{j,k}$ for requester $r_j$ to replace $r_k$ as the winner. To achieve this, $a_{j,k}$ should satisfy $a_{j,k}-\sum_{w_i\in\mathcal{C}'_j}b_i=a_k-\sum_{w_i\in\mathcal{C}'_k}b_i$, which is equivalent to $a_{j,k}=\sum_{w_i\in\mathcal{C}'_j}b_i+a_k-\sum_{w_i\in\mathcal{C}'_k}b_i$. Note that $\mathcal{C}'_1,\cdots,\mathcal{C}'_M$ denote the sets $\mathcal{C}_1,\cdots,\mathcal{C}_M$ when the specific requester $r_k$ is selected into $\mathcal{S'_R}$. If $\mathcal{C}'_j$ is not empty, the minimum value among these $a_{j,k}$'s is chosen as the payment $p_j^r$ (line \ref{line:innerforrpstart}-\ref{line:innerforrpend}); otherwise, it is further compared with $0$ (line \ref{line:ifempty}-\ref{line:looprpend}), since requester $r_j$ could win, in this case, as long as her bid is non-negative. 

Next, Algorithm \ref{al:pricing} derives the payment $p_i^w$ to each winning worker $w_i$ (line \ref{line:loopwpstart}-\ref{line:loopwpend}). Similar to line \ref{line:wdinrp}, Algorithm \ref{al:wd} is executed on the requester set $\mathcal{R}$ and worker set $\mathcal{W}$ except worker $w_i$ (line \ref{line:wdinwp}), and $\mathcal{S'_R}$ is set as the winning requester set when line \ref{line:wdinwp} stops (line \ref{line:imsetswp}). In the rest of the algorithm, we also use $\mathcal{C}'_1,\cdots,\mathcal{C}'_M$ to denote the sets $\mathcal{C}_1,\cdots,\mathcal{C}_M$ when the specific requester $r_k$ is selected into $\mathcal{S'_R}$. For each set $\mathcal{C}'_k$ such that $w_i$ belongs to $\mathcal{C}'_k$ and $r_k$ belongs to $\mathcal{S'_R}$, the algorithm calculates the maximum bid $b_{i,k}$ for worker $w_i$ to be selected as a winner at this point (line \ref{line:forwpstart}-\ref{line:loopwpend}). The calculation firstly sorts requesters in the decreasing order of their marginal social welfare, i.e., $a_j-\sum_{i:w_i\in\mathcal{C}'_j}b_i$ (line \ref{line:sort}), and finds the index $f$ of the first the requester in this order such that $w_i$ does not belong to $\mathcal{C}'_f$ (line \ref{line:firstindex}). If $r_f$ is a winning requester in $\mathcal{S'_R}$, then $b_{i,k}$ should satisfy $a_k-\big(\sum_{w_h\in\mathcal{C}'_k}b_h+b_{i,k}\big)=a_f-\sum_{w_h\in\mathcal{C}'_f}b_h$, which is equivalent to $b_{i,k}=a_k-\sum_{w_h\in\mathcal{C}'_k}b_h-\big(a_f-\sum_{w_h\in\mathcal{C}'_f}b_h\big)$; otherwise, $b_{i,k}$ should satisfy $a_k-\big(\sum_{w_h\in\mathcal{C}'_k}b_h+b_{i,k}\big)=0$, which is equivalent to $b_{i,k}=a_k-\sum_{w_h\in\mathcal{C}'_k}b_h$. Then, the maximum value among these $b_{i,k}$'s are chosen as the payment $p_i^w$ (line \ref{line:ifinSprimeR}-\ref{line:loopwpend}). Finally, Algorithm \ref{al:pricing} returns the requesters' and workers' payment profile $\mathbf{p}^r$ and $\mathbf{p}^w$ (line \ref{line:returnpricing}).


\subsubsection{Analysis of the Proposed Mechanism}
~

In this section, we prove several desirable properties of our MELON double auction, described in Algorithm \ref{al:wd} and \ref{al:pricing}. Firstly, we show its truthfulness in Theorem \ref{theo:truthfulness}. 
\begin{myTheo}\label{theo:truthfulness}
The proposed MELON double auction is truthful. 
\end{myTheo}
\begin{proof}
We prove the truthfulness of the MELON double auction by showing that it satisfies the properties of \textit{monotonicity} and \textit{critical payment}. 
\begin{itemize}[leftmargin=*]\compresslist
\item \textbf{Monotonicity.} The algorithm \texttt{FC} called by Algorithm \ref{al:wd} is independent of the requesters' and workers' bids, and winners are selected based on a decreasing order of the value $a_j-\sum_{i:w_i\in\mathcal{C}_j}b_i$. Thus, if a requester $r_j$ wins by bidding $a_j$, she will also win the auction by bidding any $a'_j>a_j$. Similarly, if a worker $w_i$ wins by bidding $b_i$, she will win the auction, as well, if her bid takes any value $b'_i<b_i$. 
\item \textbf{Critical payment.} Algorithm \ref{al:pricing} in fact pays every winning requester and worker the infimum and supremum of her bid, respectively, that can make her a winner. 
\end{itemize}

As proved in \cite{LBlumrosenAGT07}, these two properties make an auction truthful, i.e., each requester $r_j$ maximizes her utility by bidding $v_j$, and each worker $w_i$ maximizes her utility by bidding $c_i$. Therefore, the MELON double auction is truthful.
\end{proof}

Next, we show that the proposed MELON double auction satisfies individual rationality in Theorem \ref{theo:ir}. 
\begin{myTheo}\label{theo:ir}
The proposed MELON double auction is individual rational. 
\end{myTheo}
\begin{proof}
By Definition \ref{def:requesterutility} and \ref{def:workerutility}, losers of the MELON double auction receive zero utilities. From Theorem \ref{theo:truthfulness}, every winning requester $r_j$ bids $v_j$, and every winning worker $w_i$ bids $c_i$ to the platform. Moreover, they are paid, respectively, the infimum and supremum of the bid for them to win the auction. Therefore, it is guaranteed that all requesters and workers receive non-negative utilities, and thus the
proposed MELON double auction is individual rational.
\end{proof}

In Theorem \ref{theo:complexity}, we prove that the proposed MELON double auction has a polynomial-time computational complexity. 

\begin{myTheo}\label{theo:complexity}
The computational complexity of the proposed MELON double auction is $O(M^3 N+M^2 N^2)$. 
\end{myTheo}
\begin{proof}
As mentioned in Section \ref{sec:mechanism}, the algorithm \texttt{FC} (line \ref{line:feasiblecover}) in Algorithm \ref{al:wd} takes a greedy approach, and has a computational complexity of $O(N)$. Line \ref{line:Cjforstart}-\ref{line:Cjforend} of Algorithm \ref{al:wd} that find the sets $\mathcal{C}_1,\cdots,\mathcal{C}_M$ terminate at most after $MN$ steps. Next, the main loop (line \ref{line:mainloopwinnerstart}-\ref{line:mainloopwinnerend}) terminates after $M$ iterations in worst case. Within each iteration, finding the index of the requester that provides the maximum marginal social welfare (line \ref{line:maxindex}) takes $O(M)$ time, and updating the sets $\mathcal{C}_1,\cdots,\mathcal{C}_M$ takes $O(MN)$ time. Therefore, the computational complexity of the main loop is $O(MN)$, and thus, that of Algorithm \ref{al:wd} is $O(M^2 N)$ overall. 
After Algorithm \ref{al:wd}, our MELON double auction executes its pricing algorithm described by Algorithm \ref{al:pricing}, where the loop for requester pricing (line \ref{line:inipricing}-\ref{line:looprpend}) terminates in worst case after $M$ iterations. Clearly, the computational complexity of each iteration of the loop is dominated by the execution of Algorithm \ref{al:wd} in line \ref{line:wdinrp}. Therefore, the requester pricing (line \ref{line:inipricing}-\ref{line:looprpend}) in Algorithm \ref{al:pricing} takes $O(M^3 N)$ time. Following a similar method of analysis, we can conclude that the worker pricing in Algorithm \ref{al:pricing} takes $O(M^2 N^2)$ time. Hence, the computation complexity of Algorithm \ref{al:pricing}, as well as that of the overall MELON double auction is $O(M^3 N+M^2 N^2)$. 
\end{proof}

Finally, we show in Theorem \ref{theo:positivesocialwelfare} that our MELON double auction guarantees non-negative social welfare, as required. 
\begin{myTheo}\label{theo:positivesocialwelfare}
The MELON double auction guarantees non-negative social welfare. 
\end{myTheo}
\begin{proof}
Clearly, in the winner selection algorithm described by Algorithm \ref{al:wd}, a requester $r_j$ and the workers in $\mathcal{C}_j$ could be selected as winners, only if the corresponding marginal social welfare $a_j-\sum_{i:w_i\in\mathcal{C}_j}b_i$ is non-negative (line \ref{line:mainloopwinnerstart}). Thus, as the overall social welfare given by Algorithm \ref{al:wd} is the sum of the aforementioned marginal social welfare of every iteration where new winners are selected, the MELON double auction guarantees non-negative social welfare. 
\end{proof}
\section{Performance Evaluation}\label{sec:perleval}
In this section, we introduce the baseline methods, simulations settings, as well as simulation results of the performance evaluation about our proposed CENTURION framework. 
\subsection{Baseline Methods}
In our evaluation of the incentive mechanism, the first baseline auction is the \underline{M}arginal \underline{S}ocial \underline{W}elfare greedy (MSW-Greedy) double auction. As in Algorithm \ref{al:wd}, it also initializes the winner sets as $\emptyset$,  executes the algorithm $\texttt{FC}$ to obtain a feasible cover $\mathcal{C}$, and chooses from $\mathcal{C}$ the set $\mathcal{C}_j$ containing each worker $w_i$ such that $\tau_j\in\Gamma_i$ for each task $\tau_j$. Different from the MELON double auction, it sorts requesters in a decreasing order of their marginal social welfare, i.e., the value $a_j-\sum_{i:w_i\in\mathcal{C}_j}b_i$ for each requester $r_j$. Then, it selects the requester $r_j$ and the set of workers in $\mathcal{C}_j$ as winners until the marginal social welfare becomes negative. Its pricing algorithm is the same as that of the MELON double auction. Clearly, the MSW-Greedy double auction is truthful and individual rational. Another baseline auction is the one that initi\underline{A}lizes the feas\underline{I}ble cover $\mathcal{C}$ as the entire wo\underline{R}ker set $\mathcal{W}$, which we call AIR double auction. The rest of its winner selection, as well as the entire pricing algorithm is the same as those of our MELON double auction. It is easily provable that the AIR double auction is also truthful and individual rational.

Furthermore, we compare our weighted data aggregation mechanism with a mean aggregation mechanism, which outputs $+1$ as the aggregated result for a task if the mean of workers' labels about this task is non-negative, and outputs $-1$, otherwise. Another baseline aggregation mechanism that we consider is the median aggregation that takes the median of workers' labels about a task as its aggregated result. 

\subsection{Simulation Settings}

\begin{table}[h]\scriptsize\setlength{\tabcolsep}{3.6pt}
\centering
\begin{tabular}{|c|c|c|c|c|c|c|c|}
\hline
Setting&$v_j$&$c_i$&$\theta_{i,j}$&$\beta_j$&$|\Gamma_i^*|$&$N$&$M$\\
\hline\hline
I&$[10,20]$&$[5,15]$&$[0,1]$&$[0.05,0.1]$&$[15,20]$&$[90,150]$&$60$\\
\hline
II&$[10,20]$&$[5,15]$&$[0,1]$&$[0.05,0.1]$&$[15,20]$&$60$&$[20, 80]$\\
\hline
\end{tabular}
~\\~\\
\caption{Simulation settings}\label{table:setting}
\end{table}

The parameter settings in our simulation are given in Table \ref{table:setting}. Specifically, parameters $v_j$, $c_i$, $\theta_{i,j}$, $\beta_j$, and $|\Gamma_i^*|$ are sampled uniformly at random from the intervals given in Table \ref{table:setting}. The worker $w_i$'s true interested task set $\Gamma_i^*$ contains $|\Gamma_i^*|$ tasks that are randomly selected from the task set $\mathcal{T}$. In setting I, we fix the number of requesters as $60$ and vary the number of workers from $90$ to $150$, whereas we fix the number of workers as $60$ and vary the number of requesters from $20$ to $80$ in setting II. 

\subsection{Simulation Results}

\begin{figure}[h]
\begin{minipage}{.49\linewidth}
\centering
\includegraphics[width=1\textwidth]{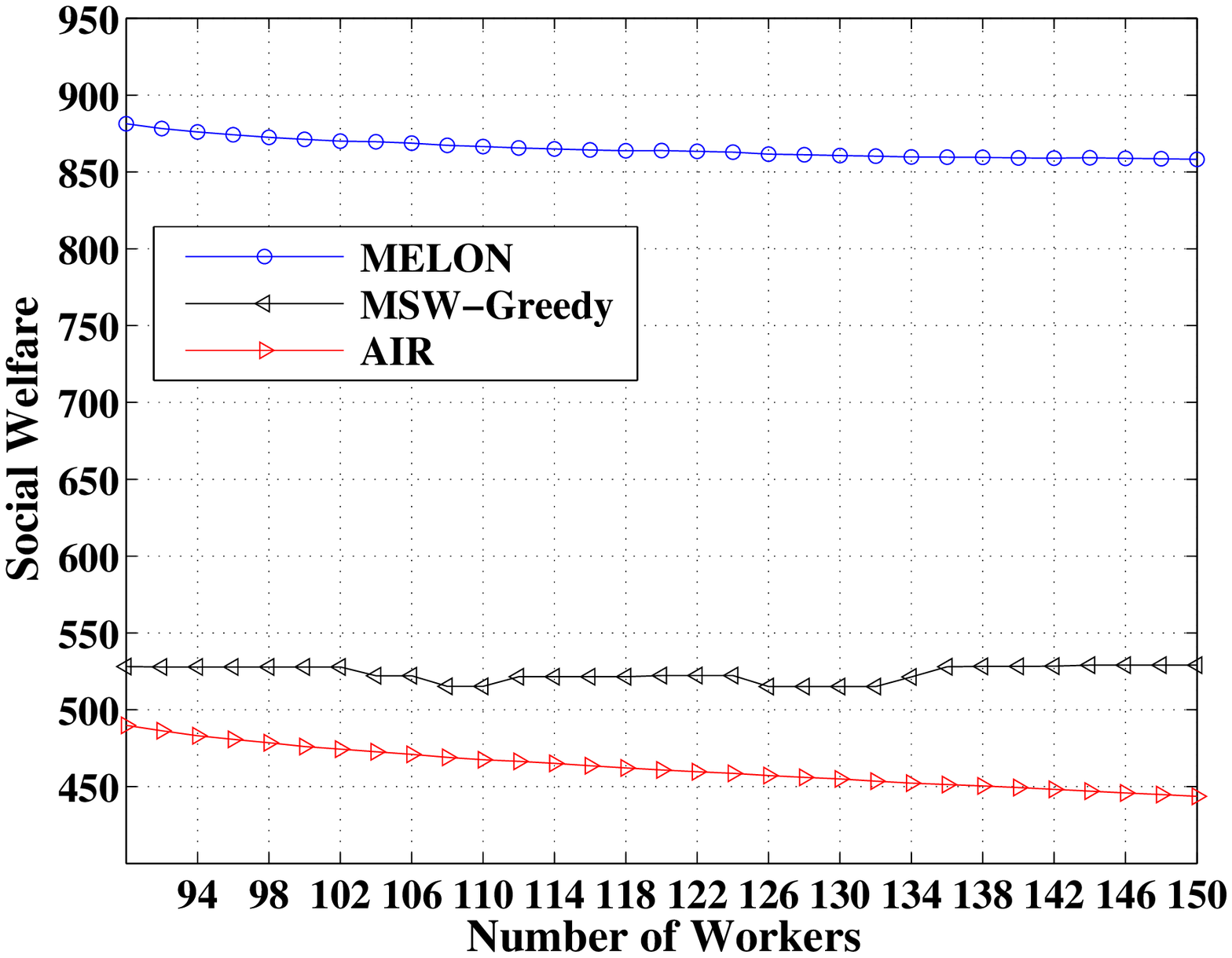}\\
\vspace{-0.2cm}
\caption{Social welfare (setting I)}
\label{fig:socialWelfareFixTask}
\end{minipage}
\begin{minipage}{.49\linewidth}
\centering
\includegraphics[width=1\textwidth]{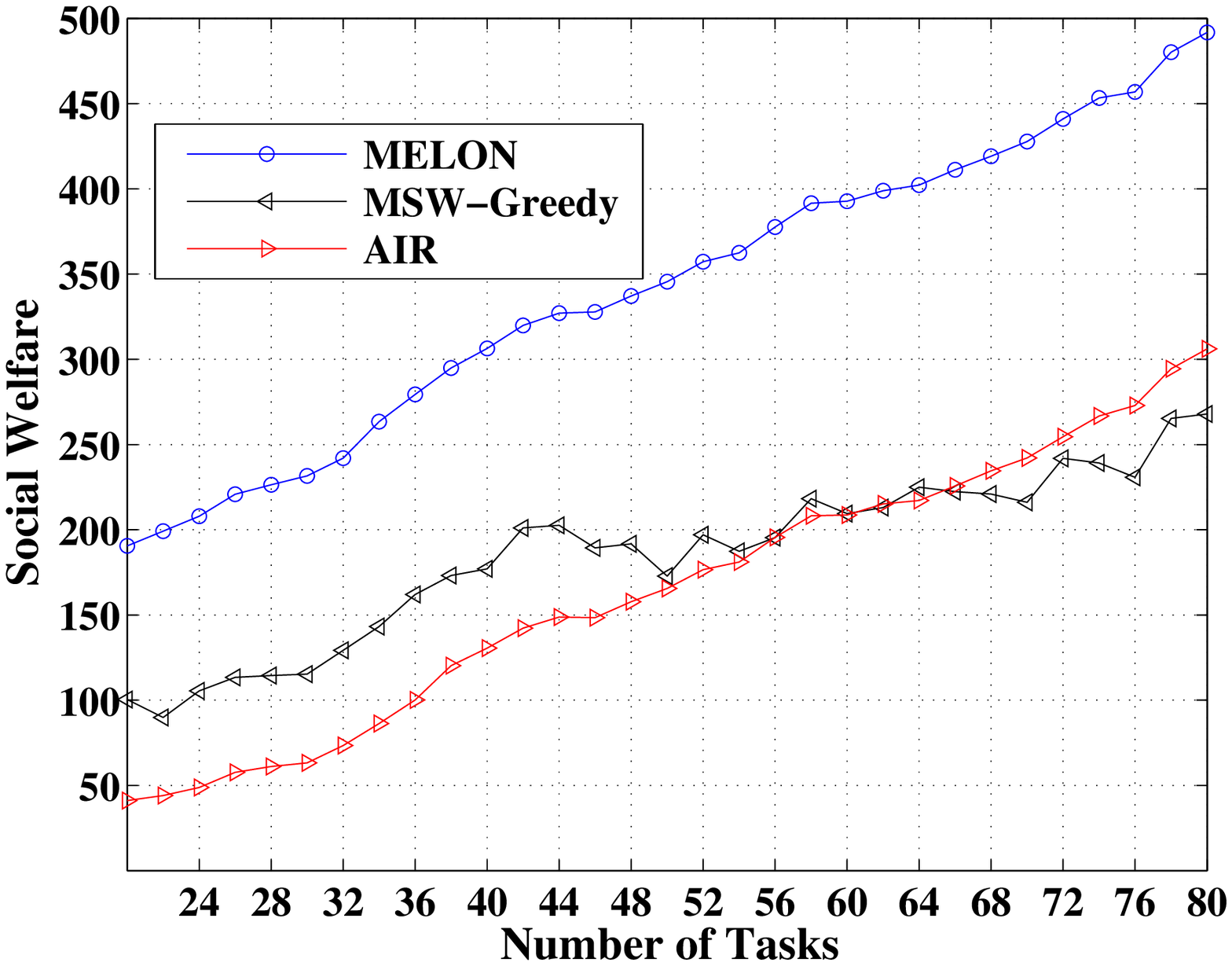}\\
\vspace{-0.2cm}
\caption{Social welfare (setting II)}
\label{fig:socialWelfareFixUser}
\end{minipage}
\end{figure}
\begin{figure}[h]
\begin{minipage}{.49\linewidth}
\centering
\includegraphics[width=1\textwidth]{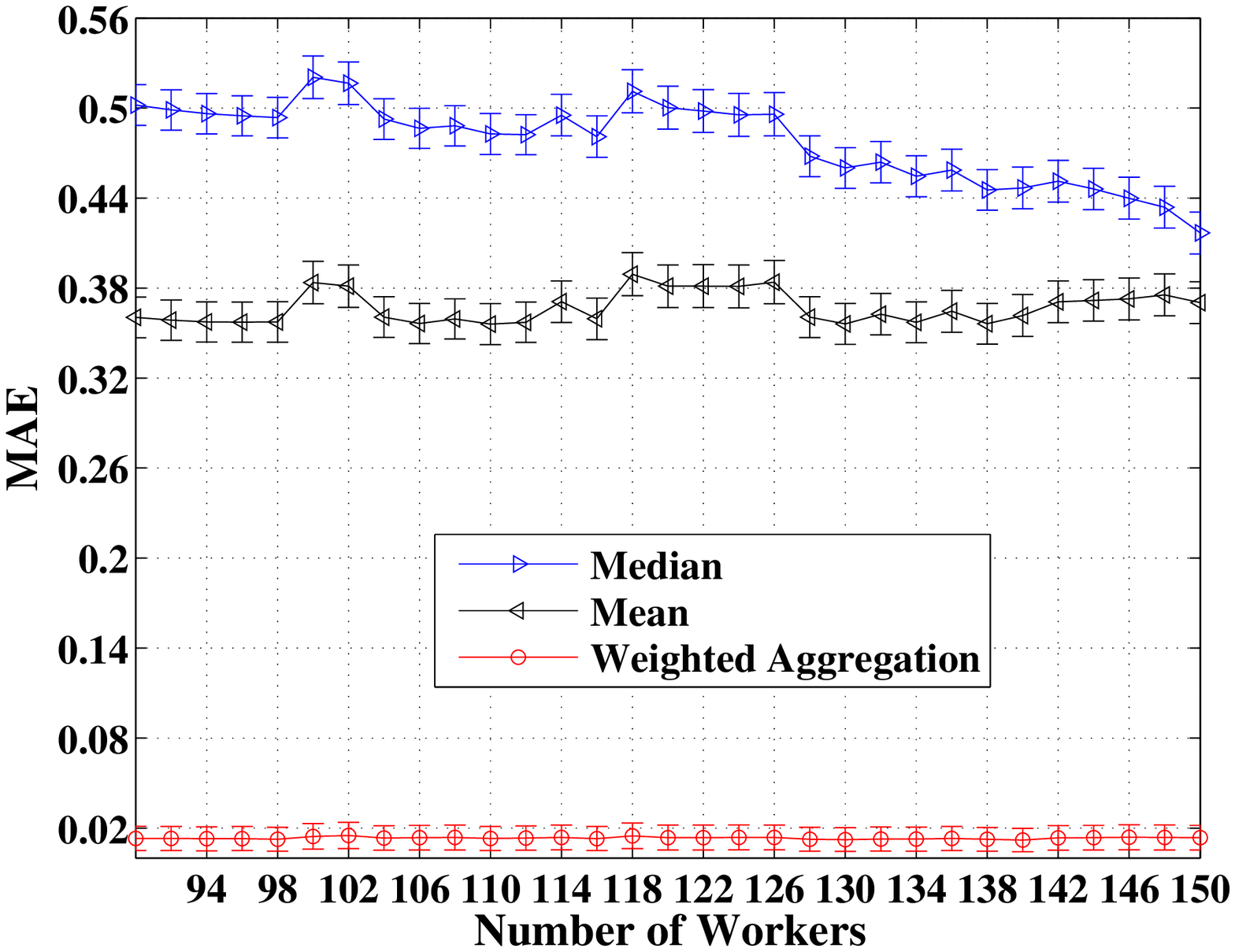}\\
\vspace{-0.2cm}
\caption{MAE (setting I)}
\label{fig:MAEFixTask}
\end{minipage}
\begin{minipage}{.49\linewidth}
\centering
\includegraphics[width=1\textwidth]{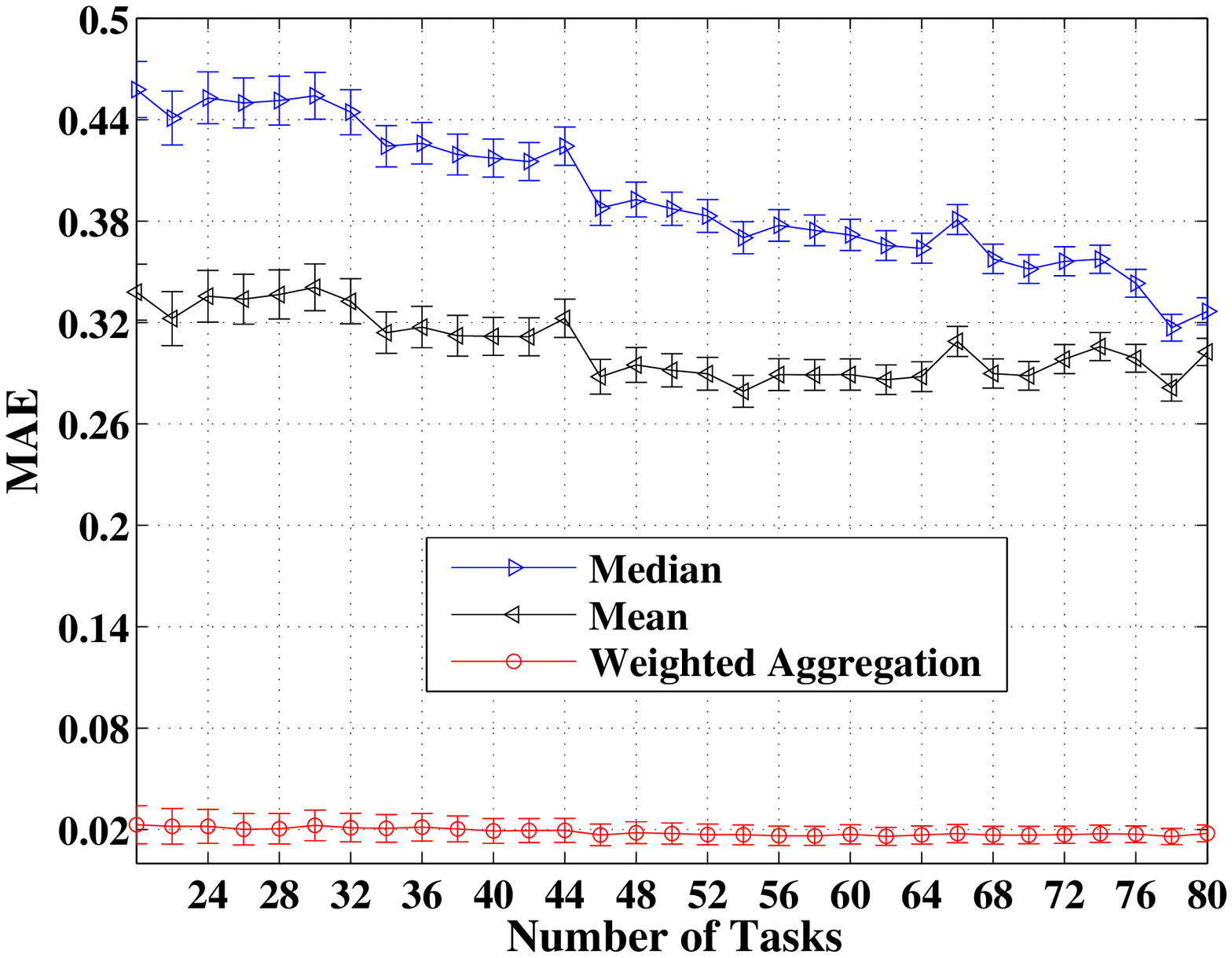}\\
\vspace{-0.2cm}
\caption{MAE (setting II)}
\label{fig:MAEFixUser}
\end{minipage}
\end{figure}
\vspace{-0.2cm}
\begin{figure}[h]
\begin{minipage}{.49\linewidth}
\centering
\includegraphics[width=1\textwidth]{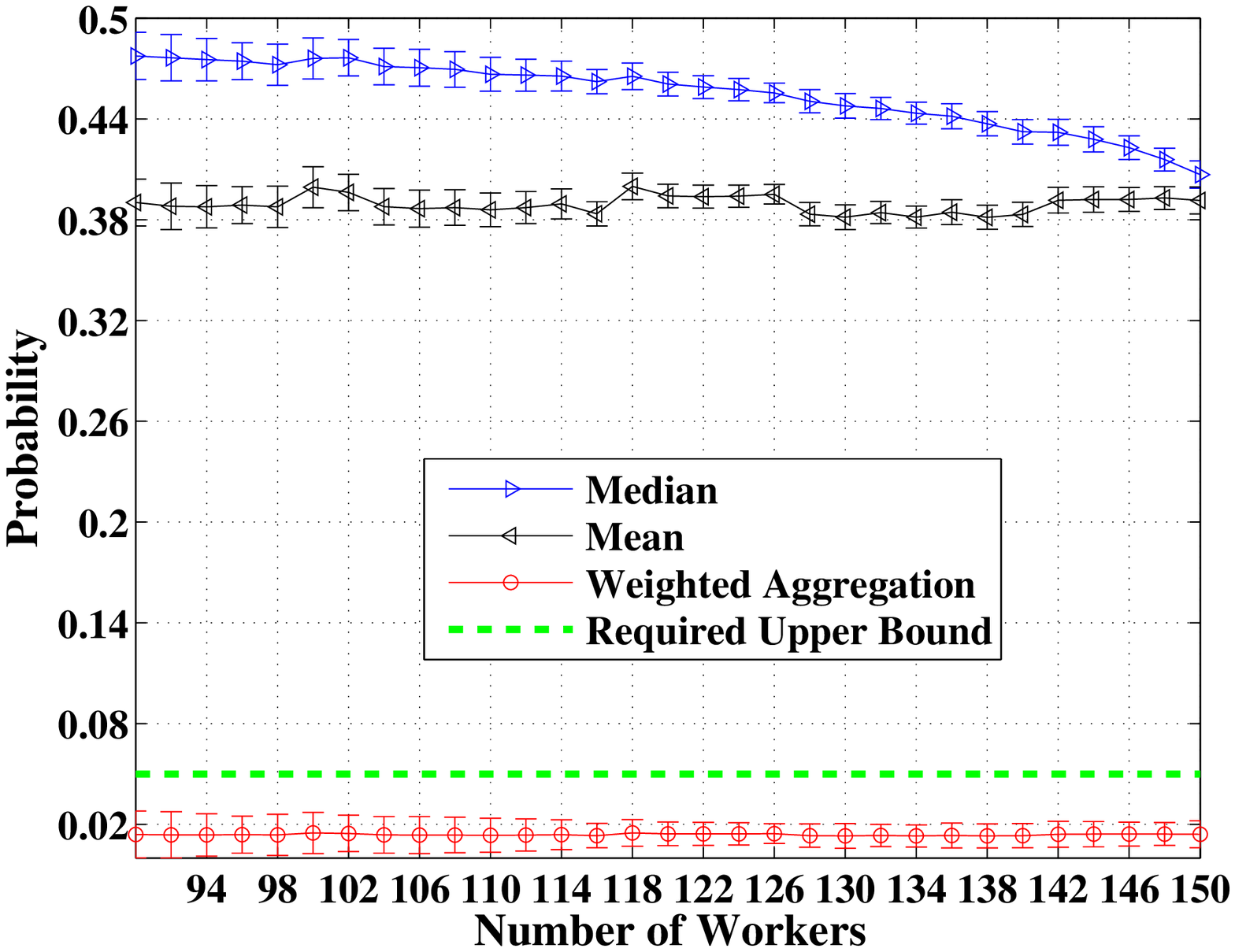}\\
\vspace{-0.2cm}
\caption{Error probability (setting I)}
\label{fig:errProbFixTask}
\end{minipage}
\begin{minipage}{.49\linewidth}
\centering
\includegraphics[width=1\textwidth]{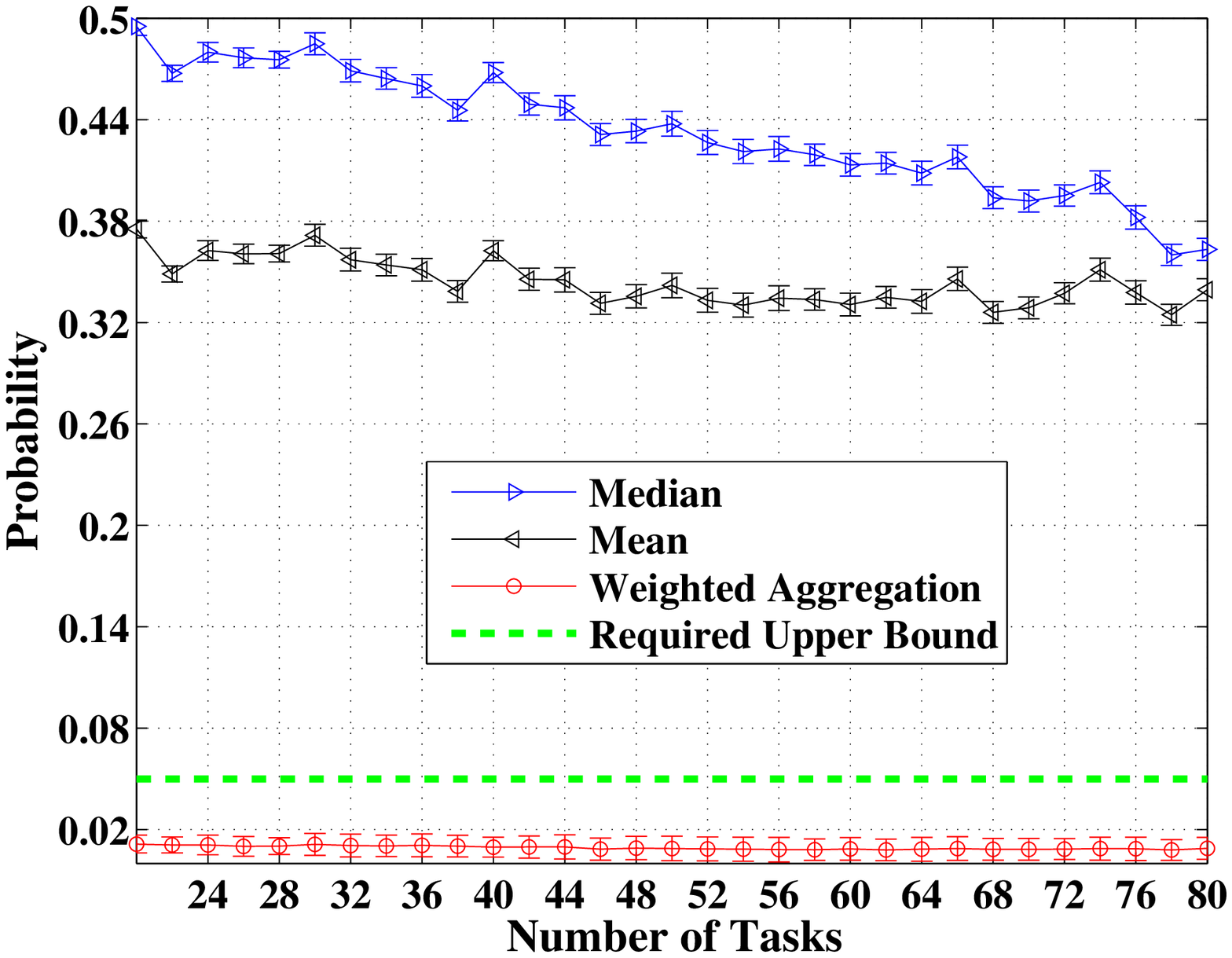}\\
\vspace{-0.2cm}
\caption{Error probability (setting II)}
\label{fig:errProbFixUser}
\end{minipage}
\end{figure}

In Figure \ref{fig:socialWelfareFixTask} and \ref{fig:socialWelfareFixUser}, we compare the social welfare generated by our MELON double auction with those of the two baseline auctions. These two figures show that our MELON double auction generates social welfare far more than the MSW-Greedy and AIR double auction under both setting I and II. 

We evaluate CENTURION's accuracy guarantee in setting I and II with a minor change of the parameter $\beta_j$, i.e., $\beta_j$ for each task $\tau_j$ is fixed as $0.05$ to simplify presentation. We compare the mean absolute error (MAE) for all tasks, which is defined as $\text{MAE}=\frac{1}{M}\sum_{j:\tau_j\in\mathcal{T}}|\hat{l}_j-l_j|$, of our weighted aggregation mechanism proposed in Algorithm \ref{al:aggregation} with those of the mean d median aggregation. The simulation for each combination of worker and requester number is repeated for $50000$ times, and we plot the means and standard deviations of the MAEs in Figure \ref{fig:MAEFixTask} and \ref{fig:MAEFixUser}. From these two figures, we observe that the MAE of our weighted aggregation mechanism is far less than those of the mean and median aggregation. Then, we show our simulation results about $\text{Pr}[|\hat{l}_j-l_j|]$, referred to as task $\tau_j$'s  error probability (EP). After $50000$ repetitions of the simulation for any given combination of worker and requester number, empirical values of the EPs are calculated, and the means and standard deviations of the empirical EPs are plotted in Figure \ref{fig:errProbFixTask} and \ref{fig:errProbFixUser}. These two figures show that the empirical EPs are less than the required upper bound $\beta_j$ and far less than those of the mean and median aggregation.

\section{Conclusion}\label{sec:conc}
In this paper, we propose CENTURION, a novel integrated framework for multi-requester MCS systems, consisting of a double auction-based incentive mechanism that stimulates the participation of both requesters and workers, and a data aggregation mechanism that aggregates workers data. Its incentive mechanism bears many desirable properties including truthfulness, individual rationality, computational efficiency, as well as non-negative social welfare, and its data aggregation mechanism generates highly accurate aggregated results. 

\begin{spacing}{1}
\bibliographystyle{IEEEtran}
\small
\bibliography{reference}
\end{spacing}

\begin{appendices}
\section{Proof of Observation \ref{theo:obs1}}

\begin{proof}
Let $\mathcal{R}$ be the set of all clauses $C_k$, such that $\mathcal{S}'$ contains any set of the form $\Gamma(C_k, \alpha)$, and assume that $|\mathcal{R}|=t$. Then the number of tasks of $E_3$ covered by $\mathcal{S}'$ is exactly $t$, and they contribute value $tZ$ to the social welfare. All remaining tasks may contribute at most $\frac{Z}{2}$ value to the social welfare. Since the cost of every set in $\mathcal{S}'$ is at least $Z$, in order for the final social welfare to be non-negative, $|\mathcal{S}'|=t$ must hold, and so $\mathcal{S}'$ contains at most one set corresponding to every clause.

Next, we prove that $|\mathcal{S}'|=m$. Let $t$ be the number of tasks of $E_1$ covered by the solution. From previous discussions, $|\mathcal{S}'|=t$. Assume for contradiction that $|\mathcal{S}'|<m$. Then the number of task of $E_2$ covered by $\mathcal{S}'$ is at most $t-1$. Therefore, the total value of all tasks covered by the solution is upper bounded by $X+tZ+(t-1)Y+10n=X+tZ+tY+6m-Y$, while the total cost of all sets in $\mathcal{S}'$ is $(Z+Y+3)t$. Therefore, the total profit is at most $X+6m-Y-3t<0$.
\end{proof}
\section{Proof of Observation \ref{theo:obs3}}

\begin{proof}
Assume for contradiction that for some $z_k\in\mathcal{O}$, both $\tau(z_k,T)$ and $\tau(z_k,F)$ are covered by $\mathcal{S}'$. Recall that there are exactly five clauses containing the variable $z_k$. Both tasks only belong to sets corresponding to these five clauses, and each such set contains exactly one of the two tasks. Since the $Q_j$ value of each such task is $5$, and for each clause $C_k$ exactly one of its corresponding set belongs to $\mathcal{S}'$, it is impossible that both $\tau(z_k,T)$ and $\tau(z_k,F)$ are covered by $\mathcal{S}'$.
\end{proof}
\section{Proof of Observation \ref{theo:obs4}}

\begin{proof}
Assume otherwise. Then $\mathcal{S}'$ covers at most $(1-\frac{\epsilon}{100})n$ tasks of $E_1$, each of which contributes 5 to the solution value, so the tasks of $E_1$ contribute at most $5n(1-\frac{\epsilon}{100})=3m-\frac{\epsilon n}{20}$ value overall. From the above discussion, the tasks of $E_2$ contribute $Ym$, those of $E_3$ contribute $Zm$, and the single task of $E_4$ contributes $X$. The total value of all tasks covered is then at most $X+3m-\frac{\epsilon n}{20}+Zm+Ym$, while the total cost of all sets is $m(Z+Y+3)$. Therefore, the social welfare is $X-\frac{\epsilon n}{20}<0$, a contradiction.
\end{proof}
\end{appendices}

\end{document}